\documentclass[10pt]{article}
\usepackage{amsmath,amsfonts,amssymb}
\usepackage{amsthm}
\usepackage{graphicx}
\usepackage{textcomp}
\usepackage{authblk}

\newtheorem{lemma}{Lemma}

\newtheorem{definition}{Definition}
\newtheorem{proposition}{Proposition}

\newtheorem{assumption}{Assumption}

\usepackage[sort]{natbib}

\usepackage{xcolor} 

\usepackage{geometry}
\begin{document}
\title{An optimal control approach of day-to-day congestion pricing for stochastic transportation networks}
\author[1]{Hemant Gehlot\thanks{Email: hgehlot@purdue.edu}}
\author[2]{Harsha Honnappa\thanks{Email: honnappa@purdue.edu}}
\author[1]{Satish V. Ukkusuri\thanks{Corresponding author. Email: sukkusur@purdue.edu. Address: Lyles School of Civil Engineering, 550 Stadium Mall Drive, West Lafayette, IN 47907-2051.}}
\affil[1]{Lyles School of Civil Engineering, Purdue University}
\affil[2]{School of Industrial Engineering, Purdue University}

\date{}
\maketitle

\begin{abstract}
Congestion pricing has become an effective instrument for traffic demand management on road networks. This paper proposes an optimal control approach for congestion pricing for day-to-day timescale that incorporates demand uncertainty and elasticity. 
Travelers make the decision to travel or not based on the experienced system travel time in the previous day and traffic managers take tolling decisions in order to minimize the average system travel time over a long time horizon. We formulate the problem as a Markov decision process (MDP) and analyze the problem to see if it satisfies conditions for conducting a satisfactory solution analysis. Such an analysis of MDPs is often dependent on the type of state space as well as on the boundedness of travel time functions. We do not constrain the travel time functions to be bounded and present an analysis centered around weighted sup-norm contractions that also holds for unbounded travel time functions. We find that the formulated MDP satisfies a set of assumptions to ensure Bellman's optimality condition. Through this result, the existence of the optimal average cost of the MDP is shown. A method based on approximate dynamic programming is proposed to resolve the implementation and computational issues of solving the control problem. Numerical results suggest that the proposed method efficiently solves the problem and produces accurate solutions.
\end{abstract}


%


\section{Introduction}
Traffic demand management through congestion pricing (CP) has become one of the most important tools to mitigate congestion. The aim of CP is to change travelers' behavior by modifying the associated travel times. In some cases, the aim is to modify the departure time by shifting some people from peak hours to off-peak times. In some cases, the goal is to discourage the use of particular routes by too many people and decrease the overall negative externalities related to congestion. The first CP was implemented in Singapore in 1975 and thereafter, it became popular in many parts of the world \citep{seik2000advanced}. There have been several works on CP in the literature. Some of these studies have focused on pricing models with an underlying static traffic assignment theory \citep{de2011traffic} whereas others have modeled dynamic congestion pricing \citep{tan2015dynamic}. Static models are simplified models that do not account for time-dependent behavior of traffic flows but are computationally less burdensome than dynamic models \citep{cheng2017urban}. 
%
%
In day-to-day modeling, CP is applied over a set of consecutive days but is only focused to a particular time period of the day. \cite{sandholm2002evolutionary} proposed a continuous CP for day-to-day route choice adjustment to guarantee efficient utilization of the system. \cite{tan2015dynamic} incorporated user
heterogeneity in day-to-day timescale CP. \cite{rambha2016dynamic} studied CP in discrete time setting considering stochastic day-to-day route choice behaviors with fixed and deterministic demand. 

Most of the aforementioned works in CP assume that demand corresponding to a particular period (like morning peak time) across different days is deterministic and inelastic. But demand variation from day to day, a critical feature of demand uncertainty, can significantly change traffic conditions and bring inefficiency to the network \citep{cantarella2016modelling}. Also, most studies do not consider dynamic tolling in a day-to-day perspective which can perform better in optimizing the traffic conditions in a day-to-day basis in comparison to static tolling \citep{rambha2016dynamic}. Some studies such as \cite{friesz2004dynamic} use continuous-time formulations to model dynamic tolling in the day-to-day scenario but continuous-time do not truly represent day-to-day setting. In addition, these works do not capitalize on the recent findings from the analysis of large scale sensing and networked data about mobility patterns at the city scale. Recent findings \citep{hasan2013spatiotemporal,zhong2016variability,yan2014universal} find regularity of travel patterns in urban cities using cell phone, smart card data and social media data. Recent studies also find the variance of travel time patterns within a narrow range thereby allowing a better characterization of uncertainty. In recent times, there has been a steep rise in the availability of big data sources that complement traditional data collection technologies like cameras and loop detectors. For example, sources like GPS and cellular data can generate frequent mobility information \citep{song2010limits}. A key big-data insight from this type of data is the repeated observations over a period of time that allows for a richer characterization of uncertainty and the mobility-demand processes. Some of the studies in real-time routing utilize the recent advancements in the technologies \citep{pi2017stochastic,gao2006optimal}. 

As mentioned before, the literature of CP models can be divided into static and dynamic modeling paradigms. Though dynamic models represent time variant flows, most of them are computationally burdensome \citep{ukkusuri2012dynamic}. Even though the academic research in traffic assignment has moved towards dynamic modeling, static models remain widely popular for strategic transport planning due to their efficiency and simplicity \citep{bliemer2012quasi}. 
Since we focus on a particular time of a day, traffic assignment across a network during the focused period of a day is assumed to be static. Once the demand in the focused period of a day is realized, traffic flows are assigned by ensuring that travel times are minimized across all the routes to ensure user optimality \citep{sheffi1985urban}. 
The travel demand during the focused period in each day is modeled as a random variable to represent the associated uncertainty. 
The elasticity in demand distribution is governed by traffic conditions (total system travel time) in the last day. That is, travelers make their travel decisions based on their memory of the traffic conditions of previous day travel. Since traffic conditions in a day is a function of the demand generated in that day plus the pricing actions taken, this problem can be formulated as a controlled Markov chain, also known as Markov decision process (MDP) \citep{puterman2014markov}.

Since demand during the focused time period of a day is a random variable, we model demand through a discrete probability distribution as it represents the number of people that travel during that period. An appropriate probability distribution to represent travel demand is Poisson distribution assuming that people independently make their decision to travel \citep{clark2005modelling,sumalee2006reliable}. Since Poisson distribution is a partially bounded distribution and the state of the system is characterized by travel demand, we formulate the problem as a countable state MDP, whose details are described in the next section. Bellman's equation, which represents the optimality condition for MDPs, is the basic entity in MDP theory and almost all existence, characterization and computational results are based on the analysis of it. Depending on the type of problem, the satisfaction of Bellman's condition by MDPs can be ensured through satisfaction of a set of assumptions. These assumptions ensure that the required contraction properties of the Bellman operator are satisfied. If these assumptions are not satisfied then one or more of the following may happen: the optimality equations may not have a solution, an optimal policy might not exist, an optimal policy may exist but it might not be stationary (time-invariant), and standard algorithms may fail to converge to the optimal solution \citep{bertsekas2012dynamic}. Apart from the type of state space, the sets of assumptions also depend on whether the cost per time step in MDP is bounded or unbounded, which in turn is dependent on whether the underlying travel time functions are bounded or unbounded. 
For example, if we consider MDPs with countable state space and bounded costs then it is sufficient to show that the Bellman operator is a contraction with respect to the \textit{unweighted} sup-norm. However, travel time functions proposed by the US Bureau of Public Roads (BPR) are not bounded \citep{boyce1981effect}, and consequently we do not assume bounded travel time functions and present an analysis using \textit{weighted sup-norm contractions} that also holds for unbounded costs \citep{bertsekas2012weighted}. 

A natural question that arises is whether the formulated MDP satisfies the assumptions to ensure the existence of the optimal average cost? In addition, while implementing a solution algorithm how should we compute and store the optimal policy for an infinite number of states? In the literature, approximate methods are proposed to implement countable state MDPs \citep{aviv1999value}, so we also develop one such approach. Also, the solution computation of MDPs becomes burdensome with the size of the problem because of the costs attached with conducting $n$-dimensional operations and subsequent storing of $n$-dimensional vectors \citep{bertsekas2012dynamic}. Hence, it is necessary devise an efficient solution algorithm for solving the MDP. Once we efficiently find the solutions of MDP, the question becomes what is the optimal policy/rule that the traffic managers should follow? Should large tolls be put during low demand periods? This paper makes a significant contribution to gain insights into the these research questions.   

In summary, this study makes the following contributions:                
\begin{itemize}
\item An average cost per time step MDP of congestion pricing is formulated for day-to-day timescale that incorporates demand stochasticity and elasticity. 
%
%
\item We analytically show that the MDP satisfies sufficient conditions to ensure Bellman's optimality conditions, implying the existence of the optimal average cost of the MDP.

\item Develop an approximate method to efficiently compute the solutions of the problem.
%
\item We conduct numerical experiments to analyze the solutions and computational performance of the approximated MDP.
\end{itemize}
This study is organized in the following manner. The next section describes the details of the problem. Section \ref{sec:solnmethod} presents the solution methodology used in this paper. Section \ref{sec:results} presents the conducted numerical results. The final section concludes the study and provides future directions.  

\section{Problem Statement}
Consider a traffic network with a single origin-destination (OD) pair (we extend the analysis to multiple OD networks in Appendix \ref{app:multipleOD}). 
We focus on a particular time period of a day, for e.g., morning peak time or evening peak time. It is assumed that traffic demand during this period of a day (or a time step in the MDP) is a random variable but the probability distribution of the demand is a function of the traffic conditions during the same period of the previous day (or the previous time step of the MDP). That is, travelers take into account the traffic conditions that were prevalent when they traveled during the same period in the previous day while making the decision to travel, resulting into an elastic demand pattern. This is an important feature of day-to-day processes where travelers adjust their decisions based on memory and learning processes \citep{watling2015model}. The system wide information of traffic conditions is available to the people through technologies like advanced traveler information systems (ATIS) and more recent advancements in big data sources like GPS, cellular data etc. The traffic assignment within a period of a day is done through static traffic assignment theory \citep{sheffi1985urban}, whose details will be provided later.

Note that we will use the terms \textit{congestion pricing} and \textit{tolling} interchangeably in this paper. The tolling agency daily decides to take tolling actions on various paths (or routes) across the OD pair.\footnote{We do not restrict the assignment of tolls to routes and the same problem can be extended to assigning the tolls to links as shown in Appendix \ref{app:overlapping}.}
At the start of each day, the agency decides how much toll to assign across different routes based on the realized traffic conditions (i.e., the travel demand in that period of the day). Such traffic information is available through the aforementioned technologies. It is assumed that these instruments cover a significant portion of the road network to provide the details of traffic conditions with good accuracy. The objective of the tolling agency is to minimize the average total system travel time over a long horizon period consisting of infinitely many days. We model the MDP as an infinite horizon problem because the evolution of traffic conditions in the day-to-day setting without actions/tolls is modeled as a Markov chain/process in the transportation literature to obtain the 'steady state distribution' of traffic flows. Thus, an infinite horizon MDP provides a 'steady state distribution' of the admissible toll pattern \citep{watling2015model,rambha2016dynamic}.
\subsection{State space}
State at time step $k$ of the MDP is the total traffic demand $x_k$ across the OD pair during that time step. As mentioned before, demand is assumed to be Poisson distributed assuming independent travel of travelers \citep{clark2005modelling,sumalee2006reliable}. So, the state space $S$ is countable and takes integer values in the interval $[0,\infty)$.   
\subsection{Action space}   
Actions are the amount of toll to be levied across the routes in the network. The action vector at time step $k$ is represented as $u_k\equiv \{u_k^1,\ldots,u_k^R\}$, where $R$ is the number of routes between the OD pair. The elements of the action vector take finite values from the set $(\tau_1,\ldots,\tau_m)$. The minimum and maximum toll values are denoted by $\tau^{min}$ and $\tau^{max}$, respectively, such that $\tau^{min}>0$ and $\tau^{max}<\infty$. The action space is denoted as $U(x_k)$, at time step $k$.
\subsection{Transition probabilities}
As mentioned before, we model stochastic demand through a Poisson distribution. The mean of the distribution in time step $k+1$ is equal to $\frac{\theta}{TSTT(x_k,u_k)}$, where $TSTT(x_k,\mu_k)$ is the total system travel time at state $x_k$ and action $u_k$ ($TSTT(x_k,\mu_k)$ is defined in detail in Section \ref{subsec:tstt}) and $\theta$ is a given positive constant. That is, the rate (or the mean in Poisson distribution) of demand in the next time step decreases if the total system travel time $(TSTT)$ of the current time step increases and vice versa, representing demand elasticity. Denote $p_{x_k x_{k+1}}(u_k)$ as the probability to go to state $x_{k+1}$ in the time step $k+1$ if the state at time step $k$ was $x_k$ and action $u_k$ was applied. The transition probability $p_{x_k x_{k+1}}(u_k)$ is given by
\begin{equation}
p_{x_k x_{k+1}}(u_k)= \frac{e^{-\frac{\theta}{TSTT(x_k,u_k)}} \left(\frac{\theta}{TSTT(x_k,u_k)}\right)^{x_{k+1}}}{x_{k+1}!}. \label{eq:poissondist}
\end{equation}

\subsection{Costs}
The cost function $g(x_k,u_k)$ is the expected cost per time step when action $u_k$ is taken at state $x_k$ and is given as follows,
\begin{equation}
g(x_k,u_k)=\sum_{x_{k+1}\in S}p_{x_k x_{k+1}}(u_k) TSTT(x_{k+1},u_k). \label{eq:expectedcost_expression}
\end{equation}

\subsection{Objective}
We define a stationary policy $\mu$ as a policy/rule of choosing actions that does not change from one time step to the next. A stationary policy $\mu$ is a \textit{feasible policy} if $\mu(x)\in U(x), \hspace{1mm} \forall x \in S$. Also, $J_\mu(x)$ is defined as the average cost per time step or the expected $TSTT$ for policy $\mu$ assuming that the system starts at state $x$, i.e., $x_0=x$, 
\begin{equation*}
J_\mu(x)= \lim\limits_{K\rightarrow\infty}\frac{1}{K}E\Bigg\{\sum_{k=0}^{K-1}g(x_k,\mu(x_k))|x_0=x\Bigg\}
\end{equation*}
where $E\{\cdot\}$ is the expected value operator. The objective is to find an optimal policy $\mu^*$ such that
\begin{equation*}
J^*(x)\equiv J_{\mu^*}(x)=\min_{\mu\in  \Pi} J_\mu (x),
\end{equation*}
where $\prod$ is the space of all feasible policies.

\subsection{Total system travel time} \label{subsec:tstt}
We now define the total system travel time $(TSTT)$. Recall that we assume the traffic assignment within a time step to be static. In the static traffic assignment literature, unbounded travel time functions of the form $c(x)=c x^a+b$ are considered \citep{boyce1981effect}. Here, $a$, $b$ and $c$ are positive constants and these travel time functions are known as BPR functions. In this paper, we consider continuous piece-wise linear approximations of BPR functions to make the mathematical analysis in the following section tractable.
Piece-wise linear approximation is a useful technique that is used to approximate many non-linear control problems for establishing useful properties of the problems \citep{sontag1981nonlinear,borrelli2005dynamic}. We consider a total of $\eta$ continuous piece-wise linear segments to approximate the original BPR function. Since the original BPR functions are monotonically increasing \citep{sheffi1985urban}, the continuous piece-wise linear approximations of BPR functions are also monotonically increasing. Each segment of the approximated function has two endpoints (or breakpoints). For the $\beta$th segment of route $r$, consider the breakpoint that has the least value of the x-coordinate. Denote the x-coordinate of this breakpoint by $x^{r,\beta}$. Then, $x^{r,1}=0, \forall r$. For route $r$, we denote the approximated travel time function to be of the form $c^r(x)=\sum_{\beta=1}^\eta (c^{r,\beta} x+\alpha^{r,\beta})z^{r,\beta}$, where $c^{r,\beta}$ and $\alpha^{r,\beta}$ are parameters corresponding to the $\beta$th approximation of route $r$. For all $\beta\neq \eta$, $z^{r,\beta}$ is a binary variable which is equal to one if $x\in [x^{r,\beta},x^{r,\beta+1})$, otherwise it is equal to zero. For $\beta = \eta$, $z^{r,\beta}$ is equal to one if $x\in [x^{r,\eta},\infty)$, otherwise it is equal to zero. 
Note that BPR functions are formulated considering no tolls, so if tolls are also present they need to be added accordingly to the travel time functions. As mentioned before, we start with the assumption that each route consists of a single link; however, we extend the analysis to incorporate assignment of tolls (on links) if routes consist of multiple links and there can be overlaps between different routes in Appendix \ref{app:overlapping}.  

We assume that the total demand $x$ (at a given time step of the MDP) gets distributed across different routes such that travel times of all the travelers are minimized, in accordance with Wardrop's principle (while distributing the flow across different routes we assume that flow is infinitesimally divisible to simplify the expression of total system travel time). This ensures that traffic flows get distributed in every time step to represent users' perspective (or user equilibrium) \citep{sheffi1985urban}.  Therefore, there are two sets of routes, 1) a \textit{used} route, where the traffic flow is positive and 2) an \textit{unused} route, where the traffic flow is zero and travel time of the route is greater than the travel time of a \textit{used} route. Denote the travel time across a used routes by $w$. Let the number of the unused routes be $\iota$ and denote the set of unused paths by $\nu$. Let $u$ denote the vector of tolls that are levied across the given time step. Then, the total system travel time $(TSTT(x,u))$ is computed as the sum of the travel times on all the paths as follows,
\begin{equation}
TSTT(x,u)=\sum_{r\in \nu} (\alpha^{r,1}+u^r)+ (R-\iota)w\label{eq:tstt1_defn}
\end{equation} 
\begin{equation}
x^r\left(\sum_{\beta=1}^\eta (c^{r,\beta} x^r+\alpha^{r,\beta})z^{r,\beta}+u^r-w\right)=0, \quad \forall r\label{eq:tstt_compl1}
\end{equation} 
\begin{equation}
\sum_{\beta=1}^\eta (c^{r,\beta} x^r+\alpha^{r,\beta})z^{r,\beta}+u^r-w\ge 0, \quad \forall r\label{eq:tstt_compl2}
\end{equation} 
\begin{equation}
x^1+\ldots+x^R=x, \label{eq:tstt_sumofflow}
\end{equation}
\begin{equation}
     x^r \ge 0, \quad \forall r  \label{eq:tstt3_2}
\end{equation}
\begin{equation}
z^{r,\beta} = \begin{cases} 1 &\text{if } x^r \in [x^{r,\beta},x^{r,\beta+1}), \quad \forall r,\beta\\
 0 &otherwise.\end{cases} \label{eq:tstt2}
\end{equation}
\begin{equation}
z^{r,\beta}\in \{0,1\}, \quad \forall r,\beta 
     \label{eq:tstt4}
\end{equation}
where $x^r$ and $u^r$ denote the traffic flow and toll value, respectively, across route $r$ at the given time step and $x^{r,\eta+1} = \infty, \forall r$. Constraint \eqref{eq:tstt1_defn} is the definitional constraint for $TSTT(x,u)$. Constraints \eqref{eq:tstt_compl1} and \eqref{eq:tstt_compl2} ensure that either the flow across a route is positive and in that case the travel time across that route is equal to the equilibrium travel time $w$, otherwise the flow across a route is zero and the travel time of that route is larger than the equilibrium travel time $w$. Constraint \eqref{eq:tstt_sumofflow} ensures that the sum of all route flows is equal to the total demand. Constraints \eqref{eq:tstt3_2}-\eqref{eq:tstt4} are definitional constraints. 
$TSTT(x,u)$ which is the sum of continuous piece-wise linear travel time functions satisfying the constraints \eqref{eq:tstt_compl1}-\eqref{eq:tstt4} is a continuous piece-wise linear function with $\xi$ piece-wise linear approximations (we argue why this should be in Appendix \ref{app:tstt_piecewiselinear}). That is,
\begin{equation}
    TSTT(x,u)=\sum_{\psi=1}^\xi \left(k^{0,\psi}x+\sum_{r }k^{r,\psi} u^r\right)\gamma^{\psi}, \label{eq:tstt_new_piecewise}
\end{equation}
where $\{k^{0,\psi}\}$ and $\{k^{r,\psi}\}$ are positive constants, $\gamma^{\psi}\in \{0,1\}, \forall \psi$ and $\sum_{\psi=1}^\xi \gamma^{\psi}=1$. 

\section{Solution methodology} \label{sec:solnmethod}
In general, the analysis of average cost per time step problems with unbounded cost per time step and denumerable states poses difficulties, and at present there is no comprehensive theory. However, there are sets of assumptions that allow a satisfactory analysis \citep{bertsekas2012dynamic}. We first present these assumptions and later show that the MDP satisfies them.
\subsection{Bellman's optimality condition}
We assume that state 0 is special in that the system has a tendency to return to it under all policies \citep{bertsekas2012dynamic}. Such a state comes under the category of \textit{recurrent} states. In particular, for any policy $\mu$ we denote $C_\mu$ as the expected cost starting from state 0 up to the first return to $0$ and $N_\mu$ as the expected number of time steps to return to state 0 starting from state 0. We start by stating the following assumption.
\begin{assumption} \label{assum:recurrence}
For every policy $\mu$, $C_{\mu}$ and $N_{\mu}$ are finite. Furthermore, $N_\mu$ is uniformly bounded over $\mu$, i.e., for some $\overline{N}>0$, we have $N_\mu<\overline{N}$ for all $\mu$.
\end{assumption}

We now introduce a positive sequence $v=\{v_0,v_1,\ldots\}$, such that 
\begin{equation*}
\inf_{i=0,1,\ldots} v_i>0,
\end{equation*}
and the weighted sup-norm
\begin{equation*}
||J||=\max_{i=0,1,\ldots} \frac{|J(i)|}{v_i}
\end{equation*}
in the space $B(S)$ of sequences $\{J(0),J(1),\ldots\}$ such that $||J||<\infty$. As stated before, weighted sup-norm contractions play an important role in the solution analysis of the MDPs that have unbounded cost per time step \citep{bertsekas2012weighted}. The following assumptions form an essential part of that analysis.
\begin{assumption}\label{assum:G}
The sequence $G=\{G_0,G_1,\ldots\}$, where 
\begin{equation*}
G_i=\max_{u\in U(i)} |g(i,u)|, \quad i=0,1,\ldots
\end{equation*}
belongs to $B(S)$ and $g(i,u)$ is the expected cost per time step when action $u$ is taken at state $i$ and is given by \eqref{eq:expectedcost_expression}. 
\end{assumption}
\begin{assumption}\label{assum:V}
The sequence $V=\{V_0,V_1,\ldots\}$, where
\begin{equation*}
V_i=\max_{u\in U(i)} \sum_{j=0}^{\infty}p_{ij}(u)v_j, \quad i=0,1,\ldots
\end{equation*}
belongs to $B(S)$ and $p_{ij}(u)$ is the transition probability to go from state $i$ to state $j$ with action $u$ and is given by \eqref{eq:poissondist}.
\end{assumption}
\begin{assumption}\label{assum:rho}
There is a scalar $\rho \in (0,1)$ and an integer $m\geq 1$ such that for all $\mu$ and $i=0,1,\ldots$, we have
\begin{equation}
\frac{\sum_{j=1}^{\infty}p^m_{ij}(\mu)v_j}{v_i}\le\rho \label{eq:rho}
\end{equation}
where $p^m_{ij}(\mu)$ is the probability of entering state $j$ at the $m$th time step given that the state at time step $0$ was $i$ and policy $\mu$ is applied.
\end{assumption}
The following proposition provides the main result.
\begin{proposition} \label{prop:bellman}
Let Assumptions \ref{assum:recurrence}, \ref{assum:G}, \ref{assum:V} and \ref{assum:rho} hold. Then the optimal average cost, denoted $\lambda^*$, is the same for all initial states and together with some sequence $h^*=\{h^*(0),h^*(1),\ldots\}$ satisfies Bellman's equation
\begin{equation}
h^*(i)+\lambda^*=\min_{u\in U(i)}\bigg[g(i,u)+\sum_{j=0}^{\infty}p_{ij}(u) h^*(j)\bigg], \quad i=0,1,\ldots \label{eq:bellman}
\end{equation}
Furthermore, if $\mu(i)$ attains the minimum in the above equation for all $i$, the stationary policy $\mu$ is optimal.
\end{proposition}
\begin{proof}
See Proposition 5.6.5 in volume 2 of \cite{bertsekas2012dynamic}.  
\end{proof}

We now show that our MDP satisfies the aforementioned assumptions under some conditions. First, we introduce the definition of a \textit{positive recurrent state} \citep{bremaud2013markov}.
\begin{definition} \label{def:recurrentstate}
A state in a Markov chain is positive recurrent if the expected number of time steps to return to the same state is finite. 
\end{definition}
We now introduce some results that will be useful in checking if the MDP satisfies Assumption \ref{assum:recurrence}.
\begin{lemma} \label{lem:foster}
(Foster's theorem) Consider an irreducible discrete-time Markov chain on a countable state space $S$ having a transition probability matrix with elements $p_{ij}$ for going from state $i$ to state $j$, where $i,j\in S$. Then, a Markov chain is positive recurrent if and only if there exists a Lyapunov function $v:S\rightarrow \mathbb{R}$, such that $v_i\ge 0, \forall i\in S$ and
\begin{equation}
\sum_{j\in S}p_{ij}v_j<\infty \quad \forall i \in F \label{eq:foster_a}
\end{equation}
\begin{equation}
\sum_{j\in S}p_{ij}v_j < v_i-\epsilon \quad \forall i \notin F \label{eq:foster_b}
\end{equation}
for some finite set $F$ and strictly positive $\epsilon$.
\end{lemma}
\begin{proof}
See Theorem 1.1, Chapter 5 of \cite{bremaud2013markov}.
\end{proof}
We now present an expression that will be useful for proving Assumption \ref{assum:recurrence} and other results in the paper.
\begin{lemma} \label{lem:expectedvalue_tstt}
$g\left(x_k,\mu_k\right)\le k^{0,m}\left(\frac{\theta}{TSTT(x_k,\mu_k)}\right)+\sum_{r }k^{r,m} \mu_k^r$
\end{lemma}
\begin{proof}
Recall that
\begin{align*}
g\left(x_k,\mu_k\right)&=\sum_{x_{k+1}=0}^{\infty}p_{x_k x_{k+1}}\left(\mu_k\right) TSTT\left(x_{k+1},\mu_k\right), \\
&=\sum_{x_{k+1}=0}^{\infty}p_{x_k x_{k+1}}\left(\mu_k\right)\sum_{\psi=1}^\xi \left(k^{0,\psi}x_{k+1}+\sum_{r }k^{r,\psi} \mu_k^r\right)\gamma^{\psi}\\
&\le\sum_{x_{k+1}=0}^{\infty}p_{x_k x_{k+1}}\left(\mu_k\right)\sum_{\psi=1}^\xi \left(k^{0,m}x_{k+1}+\sum_{r }k^{r,m} \mu_k^r\right)\gamma^{\psi} \\
&=k^{0,m}\left(\frac{\theta}{TSTT(x_k,\mu_k)}\right)+\sum_{r }k^{r,m} \mu_k^r
\end{align*}
as $E\{x_{k+1}\}=\left(\frac{\theta}{TSTT(x_k,\mu_k)}\right)$ and $\sum_{x_{k+1}=0}^{\infty}p_{x_k x_{k+1}}\left( \mu_k\right)=1$. Here, $k^{0,m}=\max_\psi \{k^{0,\psi}\}$ and $k^{r,m}=\max_\psi \{k^{r,\psi}\}$. 
\end{proof}
We now start showing that the MDP satisfies Assumption \ref{assum:recurrence}. We first show that expected recurrence times are finite for all feasible policies.
\begin{lemma}\label{lem:showfoster}
$N_{\mu}$ is finite in the MDP for all the feasible policies. 
\end{lemma}
\begin{proof}
The Foster's theorem requires the MDP to be irreducible, which is the case as all the pairs of states $i$ and $j$ are communicating since $p_{ij}(\mu)>0, \forall \mu$. Denote $v_i=k^{0,m}i+\sum_{r }k^{r,m} \tau^{max}$. 
First, we show that condition \eqref{eq:foster_a} is satisfied for all the states and feasible policies of the MDP. For state $i$ and policy $\mu$, 
\begin{align*}
\sum_{j\in S}p_{ij}(\mu)v_j=&\sum_{j\in S}p_{ij}(\mu)(k^{0,m}j+\sum_{r }k^{r,m} \tau^{max})\\
=& \frac{k^{0,m}\theta}{TSTT(i,\mu)}+\sum_r k^{r,m} \tau^{max}
\end{align*}
The last equality in the above set of equations comes after doing an analysis similar to the proof of Lemma \ref{lem:expectedvalue_tstt}. Let us compute the maximum value of $\sum_{j\in S}p_{ij}(\mu)v_j$ over all the states and policies and show that it is finite. Note that
\begin{equation}
\max_{\mu}\left(\max_{i}\left(\sum_{j\in S}p_{ij}(\mu)v_j\right)\right)=  \frac{k^{0,m}\theta}{TSTT(0,u^{min})}+\sum_r k^{r,m} \tau^{max} <\infty \label{eq:showfoster_a}
\end{equation}
where $u^{min}=\{\tau^{min},\ldots,\tau^{min}\}$ and  $TSTT(0,u^{min})$ is the total system travel time with state $0$ when the toll values on all the routes are minimum and equal to $\tau^{min}$ (this is coming from Lemma \ref{lem:tstt_mono} of Appendix \ref{app:tstt_piecewiselinear} that $TSTT$ is a monotonically increasing function). Note that the inequality in Equation \eqref{eq:showfoster_a} holds true as $\tau^{min}>0$. 

We now show that there exists a finite set of states such that all the elements not belonging to this set satisfy the condition \eqref{eq:foster_b}. We claim that for all the states $i\ge \sqrt{\frac{\theta}{k^{0,l}}}+1$, condition \eqref{eq:foster_b} is satisfied with $\epsilon = k^{0,m}\left(\sqrt{\frac{\theta}{k^{0,l}}}+1-\frac{\theta}{k^{0,l}\left(\sqrt{\frac{\theta}{k^{0,l}}}+1\right)}\right) $, where $k^{0,l}=\min_{\psi} k^{0,\psi}$. Here, $F$ is set of all integers in the interval $\bigg[0, \sqrt{\frac{\theta}{k^{0,l}}}+1\bigg)$. For a policy $\mu$, we compute 
\begin{align*}
&v_i-\sum_{j\in S}p_{ij}(\mu)v_j \\
&=k^{0,m}i+\sum_r k^{r,m}\tau^{max}-\sum_{j\in S}p_{ij}(\mu)
\left(k^{0,m}j+\sum_r k^{r,m}\tau^{max}\right) \\
&=k^{0,m}i+\sum_r k^{r,m}\tau^{max}-\frac{k^{0,m}\theta}{TSTT(i,\mu)}-\sum_r k^{r,m}\tau^{max} \\
&=k^{0,m}i-\frac{k^{0,m}\theta}{TSTT(i,\mu)} \\
&>k^{0,m}i-\frac{k^{0,m}\theta}{k^{0,l}i}.
\end{align*}
Note that the function $k^{0,m}i-\frac{k^{0,m}\theta}{k^{0,l}i}$ becomes equal to zero at $i=\sqrt{\frac{\theta}{k^{0,l}}}$ and is monotonically increasing with $i$ as the first-order derivative of the function is $k^{0,m}+\frac{k^{0,m}\theta}{k^{0,l}i^2}>0,\forall i$.  
So, for all policies and for states $i\ge \sqrt{\frac{\theta}{k^{0,l}}}+1$,
\begin{equation*}
v_i-\sum_{j\in S}p_{ij}(\mu)v_j >k^{0,m}\left(\sqrt{\frac{\theta}{k^{0,l}}}+1-\frac{\theta}{k^{0,l}\left(\sqrt{\frac{\theta}{k^{0,l}}}+1\right)}\right) 
\end{equation*}
Thus, from Lemma \ref{lem:foster}, $N_{\mu}<\infty, \forall \mu$.
\end{proof}

Before showing that the expected recurrence times $N_{\mu}$ are uniformly bounded we state the following basic limit theorem of Markov chains.
\begin{lemma} \label{lem:basiclimitthm}
Consider a recurrent irreducible aperiodic Markov chain. Let $p^k_{ii}$ be the probability of entering state $i$ at the $k$th time step given that the initial state (state at time step $0$) is $i$. Also, by convention $p^0_{ii}=1$. Let $f_{ii}^k$ be the probability of first returning to state $i$ at the $k$th time step. Thus, 
\begin{equation*}
p^k_{ii}-\sum_{l=0}^k f_{ii}^{k-l}p^l_{ii}=
\begin{cases}
      1, & \text{if}\ k=0 \\
      0, & \text{otherwise}.
    \end{cases}
\end{equation*}
Then,
\begin{equation*}
\lim_{k \rightarrow \infty} p_{ii}^k=\frac{1}{\sum_{k=0}^{\infty}kf_{ii}^k}.
\end{equation*}
\end{lemma}
\begin{proof}
See Theorem 1.2, chapter 3 of \cite{karlin2014first}.
\end{proof}
Denote $f_{ii}^k(\mu)$ as the probability of first returning to state $i$ at the $k$th time step given that the initial state was $i$ and when policy $\mu$ was applied. We first show that $f_{00}^k(\mu)$ is non-zero for all feasible policies $\mu$ of the MDP, which will be useful in showing that the expected recurrence times are uniformly bounded.
\begin{lemma} \label{lem:f00}
$f_{00}^k(\mu)$ is non-zero for all feasible policies $\mu$ of the MDP.
\end{lemma}
\begin{proof}
The probability of first returning to state $0$ after the first time step given that the initial state was $0$ and policy $\mu$ was applied is given by
\begin{equation*}
f_{00}^1(\mu)=P(x_1=0|x_0=0,\mu),
\end{equation*}
where $P(x_{k+1}=j|x_k=i,\mu)$ is the probability of entering state $j$ at time step $k+1$ given that the state at time step $k$ is $i$ and policy $\mu$ is applied and is equal to $p_{ij}(\mu)$. Similarly, we compute the first return probability at time step $2$, 
\begin{align}
f_{00}^2(\mu)&=P(x_2=0,x_1\neq 0|x_0=0,\mu)\\&
=P(x_2=0|x_1\neq 0,\mu)P(x_1\neq 0|x_0=0,\mu). \label{eq:markovpropo}
\end{align}
The equality in Equation \eqref{eq:markovpropo} comes because of the Markov property \citep{bremaud2013markov}. Following in the same way, the first return probability after $k$ time steps can be written as
\begin{equation*}
f_{00}^k(\mu)=P(x_k=0|x_{k-1}\neq 0,\mu)\prod_{m=2}^{k-1}P(x_m\neq0|x_{m-1}\neq 0,\mu)P(x_1\neq0|x_0= 0,\mu).
\end{equation*}
Note that 
\begin{align*}
P(x_k=0|x_{k-1}\neq 0,\mu) 
&=\frac{P(x_k=0,x_{k-1}\neq 0,\mu)}{P(x_{k-1}\neq 0,\mu)} \\
&=\frac{\sum_{i>0}P(x_k=0|x_{k-1}= i,\mu)P(x_{k-1}=i|\mu)}{\sum_{j>0}P(x_{k-1}=j|\mu)}\\
&=\sum_{i>0}P(x_k=0|x_{k-1}= i,\mu)\overline{P}(x_{k-1}=i|\mu),
\end{align*}
where $\overline{P}(x_{k-1}=i|\mu)=\frac{P(x_{k-1}=i|\mu)}{\sum_{j>0}P(x_{k-1}=j|\mu)}$ and $P(x_k=0|x_{k-1}= i,\mu)=e^{-\frac{\theta}{TSTT(i,\mu)}}$. Similarly, 
\begin{align*}
P(x_k\ne 0|x_{k-1}\neq 0,\mu)&=1-P(x_k=0|x_{k-1}\neq 0,\mu)
\\&=1-\sum_{i>0}P(x_k=0|x_{k-1}= i,\mu)\overline{P}(x_{k-1}=i|\mu).
\end{align*}
Also, 
\begin{align*}
P(x_1\ne 0|x_{0}= 0,\mu)&=1-P(x_1=0|x_{0}= 0,\mu)\\&=1-e^{-\frac{\theta}{TSTT(0,\mu)}}.
\end{align*}
So, it follows that 
\begin{multline*}
f^k_{00}(\mu)=\left(\sum_{i>0}e^{-\frac{\theta}{TSTT(i,\mu)}}\overline{P}(x_{k-1}=i|\mu)\right).\\ \prod_{m=2}^{k-1}\left(1-\sum_{i>0}e^{-\frac{\theta}{TSTT(i,\mu)}}\overline{P}(x_{m-1}=i|\mu)\right) \left(1-e^{-\frac{\theta}{TSTT(0,\mu)}}\right).
\end{multline*}
We now show that $f_{00}^k(\mu)$ is positive. The following holds
\begin{equation*}
1-e^{-\frac{\theta}{TSTT(0,\mu)}}>0, \forall i, \mu,
\end{equation*}
because $\frac{\theta}{TSTT(0,\mu)}>0,\mu$. Also, 
\begin{equation*}
\sum_{i>0}e^{-\frac{\theta}{TSTT(i,\mu)}}\overline{P}(x_{m-1}=i|\mu)<\sum_{i>0}\overline{P}(x_{m-1}=i|\mu),
\end{equation*}
as $e^{-\frac{\theta}{TSTT(i,\mu)}}<1, \forall i, \mu $. Note that 
\begin{equation*}
\sum_{i>0}\overline{P}(x_{m-1}=i|\mu)=\sum_{i>0}\frac{P(x_{m-1}=i|\mu)}{\left(\sum_{j>0}P(x_{m-1}=j|\mu)\right)}=1.
\end{equation*}
Therefore, 
\begin{equation*}
\prod_{m=2}^{k-1}\left(1-\sum_{i>0}e^{-\frac{\theta}{TSTT(i,\mu)}}\overline{P}(x_{m-1}=i|\mu)\right)>0.
\end{equation*}
Also, since we have shown that $P(x_k=0|x_{k-1}\ne0,\mu)>0$, so $P(x_{k-1}\ne 0,\mu)>0$ for $P(x_k=0|x_{k-1}\ne0,\mu)$ to remain defined. So, $\exists i\ne 0$, such that $\overline{P}(x_{k-1}=i|\mu)>0$. Thus, combining the fact that $e^{-x}>0\hspace{2mm} \forall x$, we get
\begin{equation*}
\sum_{i>0}e^{-\frac{\theta}{TSTT(i,\mu)}}\overline{P}(x_{k-1}=i|\mu)>0.
\end{equation*}
Therefore, $f_{00}^k(\mu)>0, \forall \mu$. 
\end{proof}
We now show that the expected recurrence times of the MDP are uniformly bounded.
\begin{lemma} \label{lem:uniformrecurrence}
For our MDP, there exists a $\overline{N}>0$, such that $N_{\mu}<\overline{N}, \forall\mu$.
\end{lemma}
\begin{proof}
We prove this lemma through the method of contradiction. Let us assume that the lemma is not true. Then, there exists a sequence of policies $\{\mu_j\}$ with $j=0,1,\ldots$ such that for this sequence the mean recurrence times $\{N_{\mu_j}\}$ form a non-decreasing sequence such that
\begin{equation*}
\lim_{j\rightarrow\infty}N_{\mu_j}=\infty.
\end{equation*}
	Note that
	\begin{equation*}
	\lim_{j\rightarrow\infty} \sum_{l=0}^k f_{00}^{k-l}(\mu_j) p_{00}^l(\mu_j) = \sum_{l=0}^k f_{00}^{k-l}(\mu_\infty) p_{00}^l(\mu_\infty),
	\end{equation*}
	and the sum is well-defined since the limits are probabilities. For a policy $\mu_j$, $N_{\mu_j}=\sum_{k=0}^{\infty}kf_{00}^k(\mu_j)$. Using the fact that $\lim_{j\rightarrow \infty} N_{\mu_j} = \infty$, it follows from Lemma \ref{lem:basiclimitthm} that
	\begin{equation*}
	\lim_{k\rightarrow\infty} \sum_{l=0}^k f_{00}^{k-l}(\mu_\infty) p_{00}^l(\mu_\infty) = 0.
	\end{equation*}
	But, this can only happen if every term in the sum is zero (observe that each term is non-negative). Therefore, it follows that
	\begin{equation*}
	f_{00}^{k-l}(\mu_\infty) p_{00}^l(\mu_\infty) = 0 \quad \forall k \geq 0, l \leq k.
	\end{equation*}
Since $p_{00}^0(\cdot) = 1$ by definition as in Lemma \ref{lem:basiclimitthm}, it must be the case that $f_{00}^{k} = 0$. But this is a contradiction of Lemma \ref{lem:f00}. Thus, our assumption that the expected recurrence times are not uniformly bounded is false, completing the proof.  
\end{proof}

\begin{lemma} \label{lem:finitecost}
For our MDP, $C_{\mu}$ is finite for all policies $\mu$. 
\end{lemma}
\begin{proof}
Denote $C'_{\mu}$ as the expected cost starting from state 0 up to the first return to 0 when the number of time steps of return is $K$. Then, 
\begin{equation*}
C'_{\mu}=\sum_{k=0}^{K-1} g\left(x_k,\mu_k\right)
\end{equation*}
such that $x_0=0,x_{K-1}=0$. From Lemma \ref{lem:expectedvalue_tstt} we get,
\begin{align*}
C'_{\mu}\le&\sum_{k=0}^{K-1}\left(k^{0,m}\left(\frac{\theta}{TSTT(x_k,\mu_k)}\right)+\sum_{r }k^{r,m} \mu_k^r\right) \\
<&\sum_{k=0}^{K-1}\left(k^{0,m}\left(\frac{\theta}{\sum_{r}k^{r,l} \mu_k^r}\right)+\sum_{r}k^{r,m} \mu_k^r\right).
\end{align*}
where $k^{r,l}=\min_{\psi} k^{r,\psi}, \forall r$. Note that $k^{0,m}\left(\frac{\theta}{\sum_{r}k^{r,l} \mu_k^r}\right)+\sum_{r}k^{r,m} \mu_k^r <\infty$ as all the possible toll values are positive and bounded (i.e. $\tau^{min}>0$ and $\tau^{max}<\infty$). Denote $\varphi^{max}$ as the upper bound of $k^{0,m}\left(\frac{\theta}{\sum_{r}k^{r,l} \mu_k^r}\right)+\sum_{r}k^{r,m} \mu_k^r$ for all policies $\mu$. Now we take expectation of $C'_{\mu}$ over the number of time steps $K$ to obtain $C_{\mu}$. So, for policy $\mu$
\begin{align*}
C_{\mu}<&E_{K} \Bigg\{\sum_{k=0}^{K-1}\left(k^{0,m}\left(\frac{\theta}{\sum_{r}k^{r,l} \mu_k^r}\right)+\sum_{r}k^{r,m} \mu_k^r\right)\Bigg\}\\
\le&E_{K}\{K\varphi^{max}\} \\
=&N_{\mu}\varphi^{max} \\
\le&\overline{N}\varphi^{max} <\infty.
\end{align*}
Thus, $C_{\mu}$ is finite for all the policies.
\end{proof}

\begin{proposition}\label{prop:recurrence}
Assumption \ref{assum:recurrence} is satisfied in the MDP. 
\end{proposition}
\begin{proof}
 Lemmas \ref{lem:showfoster}, \ref{lem:uniformrecurrence} and \ref{lem:finitecost} imply satisfaction of Assumption \ref{assum:recurrence}.
\end{proof}
We now show that the MDP satisfies Assumption \ref{assum:G}.
\begin{proposition}\label{prop:V}
There exists a sequence $v=\{v_0,v_1,\ldots\}$ that satisfies Assumption \ref{assum:G} for all the instances of the MDP.  
\end{proposition}
\begin{proof}
The expression for $G_i$ is given as 
\begin{equation*}
    G_i=\max_{u\in U(i)}|g(i,u)|
\end{equation*}
By Lemma \ref{lem:expectedvalue_tstt},
\begin{align*}
G_i&\le\max_{u\in U(i)} k^{0,m}\left(\frac{\theta}{TSTT(i,u)}\right)+\sum_{r }k^{r,m} u^r\\
&\le\max_{u\in U(i)} k^{0,m}\left(\frac{\theta}{k^{0,l}i+\sum_r k^{r,l}u^r}\right)+\sum_{r }k^{r,m} u^r\\
&\le \max \bigg\{k^{0,m}\left(\frac{\theta}{k^{0,l}i+\sum_r k^{r,l}\tau^{max}}\right)+\sum_{r }k^{r,m} \tau^{max},k^{0,m}\left(\frac{\theta}{k^{0,l}i+\sum_r k^{r,l}\tau^{min}}\right)+\sum_{r }k^{r,m} \tau^{min}\bigg\}
\end{align*}
The last inequality comes from the fact that the expression $k^{0,m}\left(\frac{\theta}{k^{0,l}i+\sum_r k^{r,l}u^r}\right)+\sum_{r }k^{r,m} u^r$ is a convex function with respect to the variables $\{u^r\}$ and thus the maximum value will lie at one of the interval limits. We now choose $v_i=k^{0,m}i+\sum_{r }k^{r,m} \tau^{max}$. Then,
\begin{equation*}
\frac{|G_i|}{v_i}\le \frac{\max \bigg\{k^{0,m}\left(\frac{\theta}{k^{0,l}i+\sum_r k^{r,l}\tau^{max}}\right)+\sum_{r }k^{r,m} \tau^{max},k^{0,m}\left(\frac{\theta}{k^{0,l}i+\sum_r k^{r,l}\tau^{min}}\right)+\sum_{r }k^{r,m} \tau^{min}\bigg\}}{k^{0,m}i+\sum_{r }k^{r,m} \tau^{max}}
\end{equation*}
Note that the maximum value of $\frac{|G_i|}{v_i}$ over $i=0,1,\ldots$, will occur at $i=0$ as the numerator is a monotonically decreasing function of $i$ and denominator is a monotonically increasing function of $i$. Thus,
\begin{equation*}
    \max_{i=0,1,\ldots} \frac{|G_i|}{v_i} \le \frac{\max \bigg\{k^{0,m}\left(\frac{\theta}{\sum_r k^{r,l}\tau^{max}}\right)+\sum_{r }k^{r,m} \tau^{max},k^{0,m}\left(\frac{\theta}{\sum_r k^{r,l}\tau^{min}}\right)+\sum_{r }k^{r,m} \tau^{min}\bigg\}}{\sum_{r }k^{r,m} \tau^{max}} <\infty
\end{equation*}
as $\tau^{min}>0$ and $\tau^{max}<\infty$. Therefore, $||J||<\infty$.
\end{proof}
Next, we show that Assumption \ref{assum:V} is satisfied by MDP. 
\begin{proposition}
There exists a sequence $v=\{v_0,v_1,\ldots\}$ that satisfies Assumption \ref{assum:V} for all the instances of the MDP.   
\end{proposition}
\begin{proof}
Set $v_i=k^{0,m}i+\sum_{r }k^{r,m} \tau^{max}$. Then, 
\begin{align*}
V_i=&\max_{u\in U(i)} \sum_{j=0}^{\infty}p_{ij}(u)v_j \\
=&\max_{u\in U(i)} \sum_{j=0}^{\infty}p_{ij}(u)\left(k^{0,m}j+\sum_{r }k^{r,m} \tau^{max}\right) 
\end{align*}
Since $\sum_{j=0}^{\infty}p_{ij}(u)j=\frac{\theta}{TSTT(i,u)}$ and $\sum_{j=0}^{\infty}p_{ij}(u)=1$, it follows that
\begin{equation*}
V_i=\max_{u\in U(i)}\left(\frac{k^{0,m}\theta}{TSTT(i,u)}+\sum_{r} k^{r,m} \tau^{max}\right) 
=\frac{k^{0,m}\theta}{TSTT(i,u^{min})}+\sum_{r} k^{r,m} \tau^{max},
\end{equation*}
where $TSTT(i,u^{min})$ is the total system travel time for state $i$ when the toll values on all the routes are minimum and equal to $\tau^{min}$. Then,
\begin{align}
E_i=\frac{|V_i|}{v_i} &= \frac{\frac{k^{0,m}\theta}{TSTT(i,u^{min})}+\sum_{r} k^{r,m} \tau^{max}}{k^{0,m}i+\sum_{r }k^{r,m} \tau^{max}}, \label{eq:VdividedV_original}\\
&\le \frac{\frac{k^{0,m}\theta}{k^{0,l}i+\sum_r k^{r,l}\tau^{min}}+\sum_{r} k^{r,m} \tau^{max}}{k^{0,m}i+\sum_{r }k^{r,m} \tau^{max}} \label{eq:VdividedV}\\
&\equiv E'_i
\end{align}
as $V_i>0 \hspace{1mm} \forall i$ holds for all the parameters of the MDP. We take an upper bound of the expression \eqref{eq:VdividedV_original} because $TSTT(i,u)$ is not a differentiable function. In order to find $\max_{i=0,1,\ldots}\frac{V_i}{v_i}$, we differentiate the expression in right-hand side (RHS) of \eqref{eq:VdividedV} with respect to $i$ to get
\begin{equation*}
\frac{\partial E'_i}{\partial i}=\frac{-k^{0,m}\theta\left(k^{0,l}(k^{0,m}i+\sum_r k^{r,m}\tau^{max})+k^{0,m}(k^{0,l}i+\sum_r k^{r,l}\tau^{min})\right)}{\left((k^{0,l}i+\sum_r k^{r,l}\tau^{min})(k^{0,m}i+\sum_{r }k^{r,m} \tau^{max})\right)^2}  -\frac{k^{0,m}\sum_{r} k^{r,m} \tau^{max}}{\left(k^{0,m}i+\sum_{r }k^{r,m} \tau^{max}\right)^2}<0, 
\end{equation*}
for all the states $i$ and parameter values. Thus, $E'_i$ is a decreasing function. Hence, the maximum value of the expression $E'_i$ occurs at $i=0$. That is,
\begin{equation*}
\max_{i=0,1,\ldots}\frac{V_i}{v_i}\le \max_{i=0,1,\ldots}E'_i=\frac{\frac{k^{0,m}\theta}{\sum_r k^{r,l}\tau^{min}}+\sum_{r} k^{r,m} \tau^{max}}{\sum_{r }k^{r,m} \tau^{max}}< \infty
\end{equation*}
as all the parameter values are finite and the minimum toll value $\tau^{min}>0$.
\end{proof}
Finally, we have the following result:
\begin{proposition}
If $\frac{k^{0,m}\theta}{\sum_{r} k^{r,l}\tau^{min}}\le \left(e^{-\left(\frac{\theta}{\sum_{r} k^{r,l}\tau^{min}}\right)}\right)\left(\sum_{r} k^{r,m}\tau^{max}\right)$ then the MDP satisfies Assumption \ref{assum:rho}.
\end{proposition}

\begin{proof}
Let $m=1$ in condition \eqref{eq:rho}. Then, left hand-side (LHS) of condition \eqref{eq:rho} becomes
\begin{equation*}
\frac{\sum_{j=1}^{\infty}p_{ij}(\mu)v_j}{v_i}.
\end{equation*}
Using the relation $\sum_{j=1}^{\infty}p_{ij}(\mu)v_j=$$\sum_{j=0}^{\infty}p_{ij}(\mu)v_j$$-p_{i0}(\mu)v_0$, we get
\begin{align*}
\frac{\sum_{j=1}^{\infty}p_{ij}(\mu)v_j}{v_i}&=\frac{\sum_{j=0}^{\infty}p_{ij}(\mu)v_j-p_{i0}(\mu)v_0}{k^{0,m}i+\sum_{r} k^{r,m}\tau^{max}} \\ &=\frac{\sum_{j=0}^{\infty}p_{ij}(\mu)(k^{0,m}j+\sum_{r} k^{r,m}\tau^{max})-p_{i0}(\mu)(k^{0,m}0+\sum_{r} k^{r,m}\tau^{max})}{k^{0,m}i+\sum_{r} k^{r,m}\tau^{max}}.
\end{align*}
Since $\sum_{j=0}^{\infty}p_{ij}(\mu)j=\frac{\theta}{TSTT(i,\mu)}$ and $\sum_{j=0}^{\infty}p_{ij}(\mu)=1$, we get
\begin{align}
&\frac{\sum_{j=0}^{\infty}p_{ij}(\mu)(k^{0,m}j+\sum_{r} k^{r,m}\tau^{max})-p_{i0}(\mu)(k^{0,m}0+\sum_{r} k^{r,m}\tau^{max})}{k^{0,m}i+\sum_{r} k^{r,m}\tau^{max}} \\&= \frac{\frac{k^{0,m}\theta}{TSTT(i,\mu)}+\sum_{r} k^{r,m}\tau^{max}}{k^{0,m}i+\sum_{r} k^{r,m}\tau^{max}} -\frac{p_{i0}(\mu)\left(\sum_{r} k^{r,m}\tau^{max}\right)}{k^{0,m}i+\sum_{r} k^{r,m}\tau^{max}} \label{eq:rhoexpresssion}
\\&< \frac{\frac{k^{0,m}\theta}{k^{0,l}i+\sum_{r} k^{r,l}\mu^r}+\sum_{r} k^{r,m}\tau^{max}}{k^{0,m}i+\sum_{r} k^{r,m}\tau^{max}} -\frac{\left(e^{-\left(\frac{\theta}{k^{0,l}i+\sum_{r} k^{r,l}\mu^r}\right)}\right)\left(\sum_{r} k^{r,m}\tau^{max}\right)}{k^{0,m}i+\sum_{r} k^{r,m}\tau^{max}}. \label{eq:rhoexpresssion_upperbound}
\end{align}
We take an upper bound of the expression \eqref{eq:rhoexpresssion} because $TSTT(i,\mu)$ is not a differentiable function. There are two set of variables in the expression \eqref{eq:rhoexpresssion_upperbound}: $i$ and $\mu$. We first find the minimum value of the expression with respect to both the variables through partial differentiation. Denote $\chi=\frac{k^{0,m}\theta}{k^{0,l}i+\sum_{r} k^{r,l}\mu^r}+\sum_{r} k^{r,m}\tau^{max}$, $\kappa=-\left(e^{-\left(\frac{\theta}{k^{0,l}i+\sum_{r} k^{r,l}\mu^r}\right)}\right)\left(\sum_{r} k^{r,m}\tau^{max}\right)$ and $\zeta=k^{0,m}i+\sum_{r} k^{r,m}\tau^{max}$. We first compute the partial derivatives of these terms with respect to $i$. Note that
\begin{equation*}
\frac{\partial \chi}{\partial i}=\frac{-k^{0,m}k^{0,l}\theta}{(k^{0,l}i+\sum_{r} k^{r,l}\mu^r)^2}<0.
\end{equation*}
Then,
\begin{align*}
\frac{\partial \kappa}{\partial i}
=&-\left(\sum_{r} k^{r,m}\tau^{max}\right)\frac{\partial e^{-\left(\frac{\theta}{k^{0,l}i+\sum_{r} k^{r,l}\mu^r}\right)}}{\partial i} \\
=&-\left(\sum_{r} k^{r,m}\tau^{max}\right)e^{-\left(\frac{\theta}{k^{0,l}i+\sum_{r} k^{r,l}\mu^r}\right)}\left(\frac{k^{0,l}\theta}{\left(k^{0,l}i+\sum_{r} k^{r,l}\mu^r\right)^2}\right)
<0.
\end{align*}
Also,
\begin{equation*}
\frac{\partial \zeta}{\partial i}=k^{0,m}>0.
\end{equation*}
Thus, maxima of $\chi$ and $\kappa$ as well as the minimum of $\zeta$ all occur when $i=0$, for all the policies and parameter settings. Therefore,
\begin{equation*}
\max_{i=0,1,\ldots} \frac{\sum_{j=1}^{\infty}p_{ij}(\mu)v_j}{v_i}< \frac{\frac{k^{0,m}\theta}{\sum_{r} k^{r,l}\mu^r}+\sum_{r} k^{r,m}\tau^{max}}{\sum_{r} k^{r,m}\tau^{max}} -\frac{\left(e^{-\left(\frac{\theta}{\sum_{r} k^{r,l}\mu^r}\right)}\right)\left(\sum_{r} k^{r,m}\tau^{max}\right)}{\sum_{r} k^{r,m}\tau^{max}}
\end{equation*}
as $\zeta=k^{0,m}i+\sum_rk^{r,m}\tau^{max}$ is positive. We now compute the maximum of the expression in RHS of Equation \eqref{eq:rhoexpresssion_upperbound} with respect to $\mu$. Taking partial derivatives of the RHS expression in Equation \eqref{eq:rhoexpresssion_upperbound} with respect to $\mu^r \hspace{1mm}\forall r$, we get
\begin{equation*}
-\frac{\frac{k^{0,m}k^{r,l}\theta}{(k^{0,l}i+\sum_r k^{r,l}\mu^r)^2}}{k^{0,m}i+\sum_{r} k^{r,m}\tau^{max}}-\frac{e^{-\left(\frac{\theta}{k^{0,l}i+\sum_{r} k^{r,l}\mu^r}\right)}\left(\frac{\theta k^{r,l}}{(k^{0,l}i+\sum_r k^{r,l}\mu^r)^2}\right)\left(\sum_{r} k^{r,m}\tau^{max}\right)}{k^{0,m}i+\sum_{r}k^{r,m}\tau^{max}} <0.
\end{equation*}
That is, the expression of Equation \eqref{eq:rhoexpresssion_upperbound} is a decreasing function with respect to $\mu^r, \forall r$. Therefore, the maximum of the expression occurs at $\mu=\{\tau^{min},\tau^{min},\ldots,\tau^{min}\}$, for all states and parameter settings. 

Thus, by setting $i=0$ and $\mu=u^{min}$ we bound the expression in Equation \eqref{eq:rhoexpresssion} as follows:
\begin{equation*}
\frac{\frac{k^{0,m}\theta}{TSTT(i,\mu)}+\sum_{r} k^{r,m}\tau^{max}}{k^{0,m}i+\sum_{r} k^{r,m}\tau^{max}} -\frac{p_{i0}(\mu)\left(\sum_{r} k^{r,m}\tau^{max}\right)}{k^{0,m}i+\sum_{r} k^{r,m}\tau^{max}}
<1+\frac{\frac{k^{0,m}\theta}{\sum_{r} k^{r,l}\tau^{min}}}{\sum_{r} k^{r,m}\tau^{max}} -\frac{\left(e^{-\left(\frac{\theta}{\sum_{r} k^{r,l}\tau^{min}}\right)}\right)\left(\sum_{r} k^{r,m}\tau^{max}\right)}{\sum_{r} k^{r,m}\tau^{max}}.
\end{equation*}
Finally, we derive a sufficient condition to ensure that \eqref{eq:rho} is satisfied. Since the supremum of $\rho$ in condition \eqref{eq:rho} is 1, we ensure that the expression in the last condition is less than or equal to 1 (as $\rho \in (0,1)$). That is, 
\begin{equation*}
1+\frac{\frac{k^{0,m}\theta}{\sum_{r} k^{r,l}\tau^{min}}}{\sum_{r} k^{r,m}\tau^{max}} -\frac{\left(e^{-\left(\frac{\theta}{\sum_{r} k^{r,l}\tau^{min}}\right)}\right)\left(\sum_{r} k^{r,m}\tau^{max}\right)}{\sum_{r} k^{r,m}\tau^{max}}\le 1
\end{equation*}
Or,
\begin{equation*}
\frac{k^{0,m}\theta}{\sum_{r} k^{r,l}\tau^{min}}\le \left(e^{-\left(\frac{\theta}{\sum_{r} k^{r,l}\tau^{min}}\right)}\right)\left(\sum_{r} k^{r,m}\tau^{max}\right)
\end{equation*}
\end{proof}
Finally, it follows that Proposition \ref{prop:bellman} holds for the MDP. 
\subsection{Implementation issues in solving the problem}\label{sec:implementation}
We now focus on implementing a solution algorithm to solve the MDP. A few studies have established the convergence of algorithms like value iteration and policy iteration for countable state MDPs \citep{meyn1997policy,sennott1991value,aviv1999value}. But the countable state MDP, while being theoretically appealing, faces several implementation issues. First, the infinite sums to the right of Equation \eqref{eq:bellman} cannot be evaluated on a term by term basis, nor can the function $h^*(i)$ be evaluated and stored for infinitely many values \citep{aviv1999value}. Therefore, approximate methods need to be devised to find the solutions of countable state MDPs.  

One approach to approximate a countable state MDP is by solving the finite state problem obtained through truncation of the state space. Existing results show that as the size of the approximating MDP increases, its cost per time step and, under some conditions its optimal policies approach those of the original, countable MDP \citep{cavazos1986finite,altman1994denumerable}. So, for solving the problem we truncate the state space $S$ to a finite interval $[0,x^{max}]$. Once we solve the finite state MDP we can extrapolate the optimal policies for the states $S\backslash [0,x^{max}]$. The truncation of the state space implies the state transition probabilities are now governed by a truncated Poisson distribution. We denote the transition probability from state $x_k$ to state $x_{k+1}$ for action $u_k$ using truncated Poisson distribution as $q_{x_k x_{k+1}}(u_k)$, where 
\begin{equation*}
q_{x_k x_{k+1}}(u_k)= \frac{\frac{e^{-\frac{\theta}{TSTT(x_k,u_k)}} \left(\frac{\theta}{TSTT(x_k,u_k)}\right)^{x_{k+1}}}{x_{k+1}!}}{\sum_{X=0}^{X=x^{max}}\frac{e^{-\frac{\theta}{TSTT(x_k,u_k)}} \left(\frac{\theta}{TSTT(x_k,u_k)}\right)^X}{X!}}.
\end{equation*}
The action space, costs and objective function remain the same as before. We now present the solution algorithm that we use to solve the truncated MDP.  
\subsection{Solution algorithm}\label{subsec:PI}
We use policy iteration to solve for the optimal policy of the infinite horizon finite state MDP. Policy iteration algorithm starts with a stationary policy $\mu^0$, and iteratively generates a sequence of new policies $\mu^1,\mu^2,\ldots$ as follows:
\begin{enumerate}
\item (Policy evaluation) Given a policy $\mu^m$ at iteration $m$, we compute average and differential costs $\lambda^m$ and $h^m(i)$ satisfying
\begin{equation}
h^m(i)+\lambda^m=g(i,\mu^m(i))+\sum_{j=0}^{x^{max}}q_{ij}(\mu^m(i)) h^m(j), \quad i=0,\ldots,x^{max} \label{eq:policy_eval}
\end{equation}
and $h^m(x^{max})=0$. The condition $h^m(x^{max})=0$ assumes without loss of generality that the state $x^{max}$ is recurrent and there exists a unique solution to \eqref{eq:policy_eval}.
\item (Policy improvement) We find a stationary policy $\mu^{m+1}$, where for all $i$, $\mu^{m+1}(i)$ is such that
\begin{equation*}
g(i,\mu^{m+1}(i))+\sum_{i=0}^{x^{max}}q_{ij}(\mu^{m+1}(i))h^m (j)=\min_{u\in U(i)}\left[g(i,u)+\sum_{i=0}^{x^{max}}q_{ij}(u)h^m(j)\right].
\end{equation*}
\item (Termination check) If $\lambda^{m+1}=\lambda^m$ and $h^{m+1}(i)=h^m (i)$ for all $i$, the algorithm terminates; otherwise, we go to Step 1 with $\mu^{m+1}$ replacing $\mu^m$.
\end{enumerate}
\subsection{Aggregated problem} \label{sec:aggregated}
We have seen that dynamic programming (DP) algorithms like Policy iteration can be used to solve the finite state MDP. However, the computational requirements associated with the exact solution of even finite state MDPs are overwhelming. It is known that DP algorithms can be numerically applied only if the size of state space is relatively small. Computing the exact solution is generally difficult and possibly intractable for large problems due to the widely known \textit{curse of dimensionality} \citep{bertsekas2012dynamic}. Thus, it is necessary to devise an approximate method to tackle this issue.
  
We use state aggregation to overcome the large state space. Before performing the aggregation, we assume that the state space is continuous by using the normal approximation to the Poisson distribution. So, the probability density function of going to the state $x_{k+1}$ at time step $k+1$ if the state in the previous time step was $x_k$ and the action $u_k$ was taken is given by
\begin{equation*}
f_{x_k x_{k+1}}\left(u_k\right)=\frac{1}{\sqrt{2\pi\sigma^2}}e^{-\frac{\left(x_{k+1}-\upsilon\right)^2}{2\sigma^2}},
\end{equation*} 
where $\upsilon=\frac{\theta}{TSTT(x_k,u_k)}, \sigma^2=\frac{\theta}{TSTT(x_k,u_k)}$. 
\subsubsection{State space}
We construct a finite number of states from the continuous state space by aggregating states using intervals. We define the aggregated state space set $S'=\Big\{\left[0,\frac{x^{max}}{N}\right],\left[\frac{x^{max}}{N},\frac{2x^{max}}{N}\right],\ldots,\left[\frac{(N-1)x^{max}}{N},x^{max}\right]\Big\}$ by dividing $x^{max}$ into $N$ intervals. For an aggregated interval/state $y \in S'$, we approximate the values in the interval by its center, $y^c $.  

\subsubsection{Transition probabilities}
We compute the state transition probability of going from an aggregated state $X$ to another aggregated state $Y$ by computing the probability of going from the center of state $X$ to state $Y$. We define 
\begin{equation*}
p'_{y_k^c y_{k+1}^c}(u_k)=\int_{y_{k+1}^c - \frac{x^{max}}{2N}}^{y_{k+1}^c + \frac{x^{max}}{2N}}\frac{1}{\sqrt{2\pi\sigma^2}}e^{-\frac{\left(x-\upsilon\right)^2}{2\sigma^2}} dx,
\end{equation*} 
where $\upsilon=\frac{\theta}{TSTT(y_k^c,u_k)}, \sigma^2=\frac{\theta}{TSTT(y_k^c,u_k)}$. Since we truncate the state space in the interval $[0,x^{max}]$, it may happen that the sum of probabilities of going from state $X$ to all the aggregated states might not add up to one. Therefore, we now normalize these values by setting
\begin{equation*}
p_{y_k^c y_{k+1}^c}(u_k)=\frac{p'_{y_k^c y_{k+1}^c}(u_k)}{\sum_{\forall y^c \in S'} p'_{y_k^c y^c}(u_k)}.
\end{equation*}

Note that the action space for the approximate dynamic programming remains the same as before. The expressions for costs also remain the same as before and only the original states are replaced by the centers of the aggregated states. Thus, the policy iteration algorithm can be used similarly as in Section \ref{subsec:PI} by appropriately using the state space and transition probabilities of the aggregated problem. 

\section{Numerical results} \label{sec:results}
Consider a single origin-destination (OD) pair with two routes (each containing a single link) as shown in Figure \ref{fig:network}. In the appendices, we show that networks with complex topologies can be reduced to single OD pair networks with non-overlapping routes. So, we present the results for the simplest setting. The travel time parameters are given as follows, $c=[1,2]$ and $b=[0.5,1]$ where the first and second entries correspond to routes 1 and 2, respectively. The power coefficient $a$ is set to 4 according to the standard BPR functions \citep{boyce1981effect}. There are three possible toll values: 2, 3 and 4 units. The parameter $\theta$ of the Poisson distribution is set to 100. The maximum value of demand per unit time $x^{max}$ is set to 15. We term this problem as the \textit{original} problem in this paper. The results that we will present either correspond to the \textit{original} problem or modifications of it.    
\begin{figure}[h]
\centering
	\includegraphics[width= 0.25\textwidth]{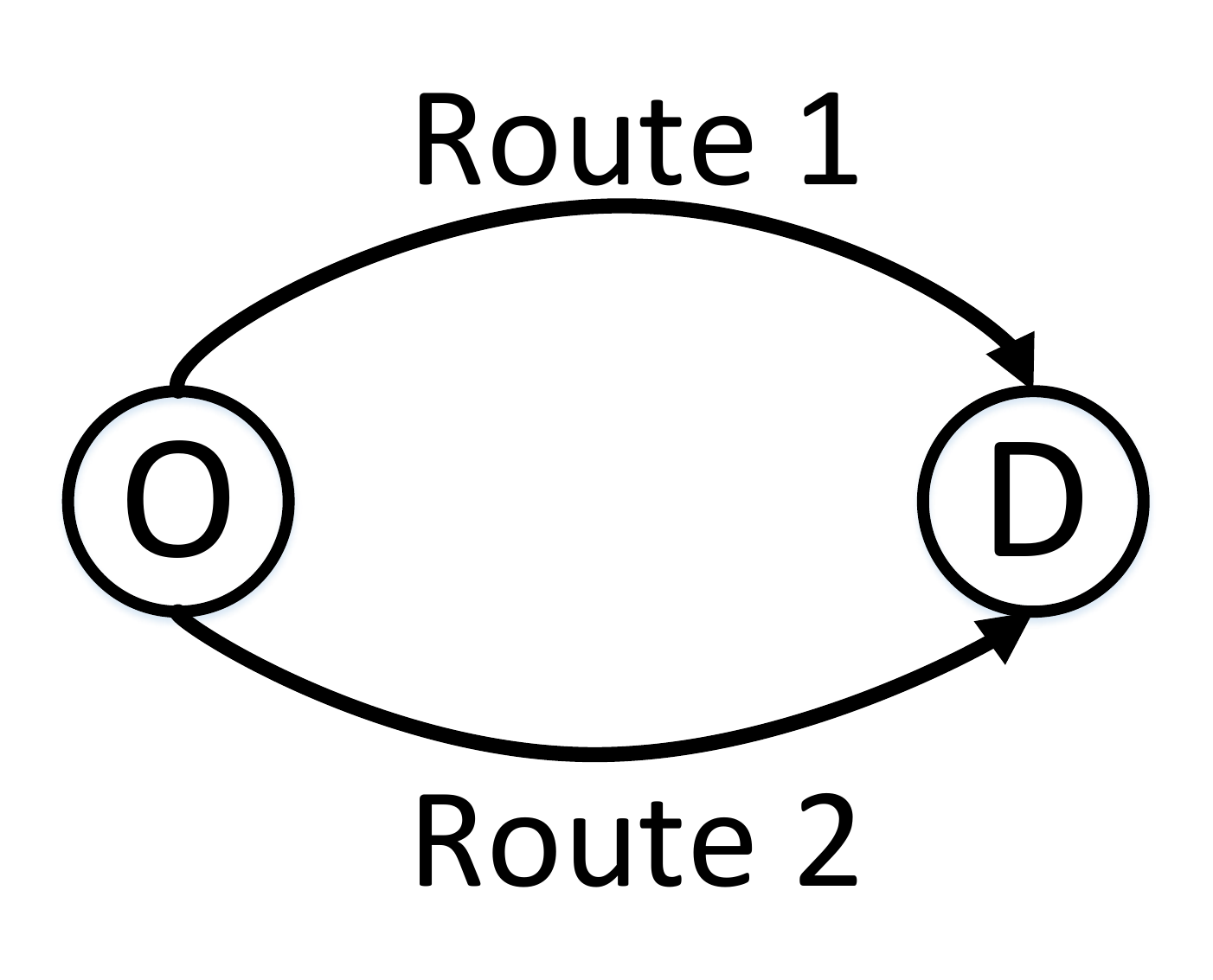}
	\caption{Test network}
	\label{fig:network}
\end{figure}
\begin{figure}[h]
\centering
	\includegraphics[width= 0.45\textwidth]{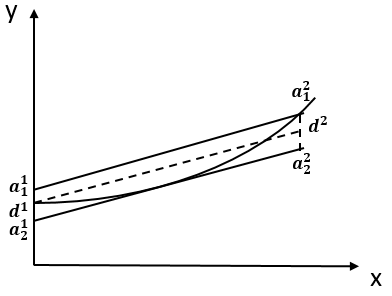}
	\caption{Construction of the first segment of the approximated continuous piece-wise linear function}
	\label{fig:piecewise}
\end{figure}

We first discuss the method that we use to generate approximations of the BPR functions, though any method that can generate continuous piece-wise linear approximations of the original functions is valid. The objective of these approximations is to ensure that the difference (measured in terms of a norm) between the approximated curve and the original curve does not exceeds a given value. One of the possible norms is the maximum norm that ensures that the maximum difference between the original curve and the approximated curve does not exceeds a given value. \cite{tomek1974two} introduced a method to ensure that the approximated curve satisfies the aforementioned criteria. The first segment of the approximated curve is the segment whose one breakpoint lies on the y-axis. Let $d^j=(x^j,y^j)$ be the point with the least x-coordinate in the $j$th segment of the approximated curve. The basic idea starts from the construction of two parallel line segments passing through the points $a_1^j=(x^j,y^j+\epsilon)$ and $a_2^j=(x^j,y^j-\epsilon)$ that differ by $2\epsilon$ in the y direction for the $j$th segment of the piece-wise continuous function. The line segments are extended until it is not possible to bound the original curve between the two parallel line segments. Figure \ref{fig:piecewise} illustrates the construction of the two segments for $j=1$. Since the domain corresponding to the original BPR function goes to infinity, we construct the first $\eta-1$ segments such that maximum error between the original and the approximated curves is $2\epsilon$ but this rule is not followed for the last segment that extends to positive infinity. We use the following algorithm to generate piece-wise linear approximations (the original algorithm presented by \cite{tomek1974two} can be substantially simplified to the following algorithm because BPR functions are monotonically increasing):  
\begin{enumerate}
    \item Set $j=1$, $x^1=0$ and $y^1=b$, where $b$ is the constant corresponding to the BPR function that needs to be approximated. Denote $d^j=(x^j,y^j)$ and $a_1^j=(x^j,y^j+\epsilon)$, $a_2^j=(x^j,y^j-\epsilon)$.
    \item Construct a tangent starting from the point $a_2^j$ that intersects the curve of the original BPR function.
    \item Construct a line segment starting from the point $a_1^j$ that is parallel to the previously constructed tangent. Denote the point that represents the intersection of the line segment with the original BPR function by $a_1^{j+1}$. Then, the points $d^{j+1}$ and $a_2^{j+1}$ are at distances $\epsilon$ and $2\epsilon$ below the point $a_1^{j+1}$ in y-direction, respectively.
    \item Set $j=j+1$. If $j=\eta$, then go to the next step. Otherwise, go to step 2.
    \item Construct a tangent starting from the point $a_2^{\eta}$ that intersects the curve of the original BPR function. Draw a ray  starting from the point $d^{\eta}$ that is parallel to the previously constructed tangent. The constructed ray represents the $\eta$th approximated segment.
\end{enumerate}
In the \textit{original} problem, we set $\epsilon=1$ and $\eta=4$. We discuss about the effect of number of approximation segments (i.e., $\eta$) on the solution accuracy and computation later.

\subsection{Solutions} \label{subsec:solution}
In this section, we focus on the solutions obtained by solving \textit{original} problem and the variants of it. We coded the problems in MATLAB \citep{coleman}, and solved the computation of \textit{TSTT} through \eqref{eq:tstt1_defn}-\eqref{eq:tstt4} using a TOMLAB solver \citep{holmstrom2010user}. We first present the solutions obtained from solving the \textit{original} problem as described in Section \ref{subsec:PI} (i.e., without aggregation of states). Figure \ref{fig:soln} provides the optimal toll values for the two routes. The dashed line presents the optimal toll values for route 1 and the solid line provides the toll values for route 2 (in this example, the optimal toll values are equal for both the routes so only the solid line is visible but that may not be always true as can be seen in later results where some of the parameters of the \textit{original} problem are varied). It can be seen that for both the routes the optimal value of toll is larger for smaller state values as compared to larger state values. This is possibly because when the demand is low, putting large toll value reduces the probability of a large demand in the next time step by Equation \eqref{eq:poissondist}. This switching of optimal policy is also dependent on the parameters $\theta$ and $\{k^{0,\psi}\},\{k^{1,\psi}\},\ldots,\{k^{R,\psi}\}$. To better understand this, we focus on the expected cost per time step $g(x,u)$. Since $TSTT(x,u)$ is a piece-wise linear function, it is not possible to partially differentiate it and therefore it is not possible to partially differentiate $g(x,u)$. Therefore, we analyze the upper and lower bounds of $g(x,u)$. Note that $g\left(x,u\right)\le k^{0,m}\left(\frac{\theta}{TSTT(x,u)}\right)+\sum_{r }k^{r,m} u^r\le k^{0,m}\left(\frac{\theta}{k^{0,l}x+\sum_{r }k^{r,l} u^r}\right)+\sum_{r }k^{r,m} u^r$ by Lemma \ref{lem:expectedvalue_tstt}. We partially differentiate the expression $k^{0,m}\left(\frac{\theta}{k^{0,l}x+\sum_{r }k^{r,l} u^r}\right)+\sum_{r }k^{r,m} u^r$ with respect to $u^r, \forall r $. The derivative of the expression with respect to $u^r$ is equal to zero when (we also checked the second order conditions to verify that the first order conditions represent the minimum)
\begin{equation*}
 \sum_{r \in R}k^{r,l} u^r = \sqrt{\frac{k^{0,m}k^{r,l}\theta}{k^{r,m}}}-k^{0,l}x.    
\end{equation*}
Consider a network of two routes for simplicity. Then, the partial derivative of the upper bound of $g(x,u)$ with $u^1$ is  
\begin{equation}
k^{1,l} u^1+ k^{2,l} u^2= \sqrt{\frac{k^{0,m}k^{1,l}\theta}{k^{1,m}}}-k^{0,l}x, \label{eq:upperbound_g_optimality}
\end{equation}
and the partial derivative of the same expression with $u^2$ is
\begin{equation}
k^{1,l} u^1+ k^{2,l} u^2= \sqrt{\frac{k^{0,m}k^{2,l}\theta}{k^{2,m}}}-k^{0,l}x. \label{eq:upperbound_g_optimality_2}
\end{equation}
\begin{figure}[h]
\centering
	\includegraphics[width= 0.75\textwidth]{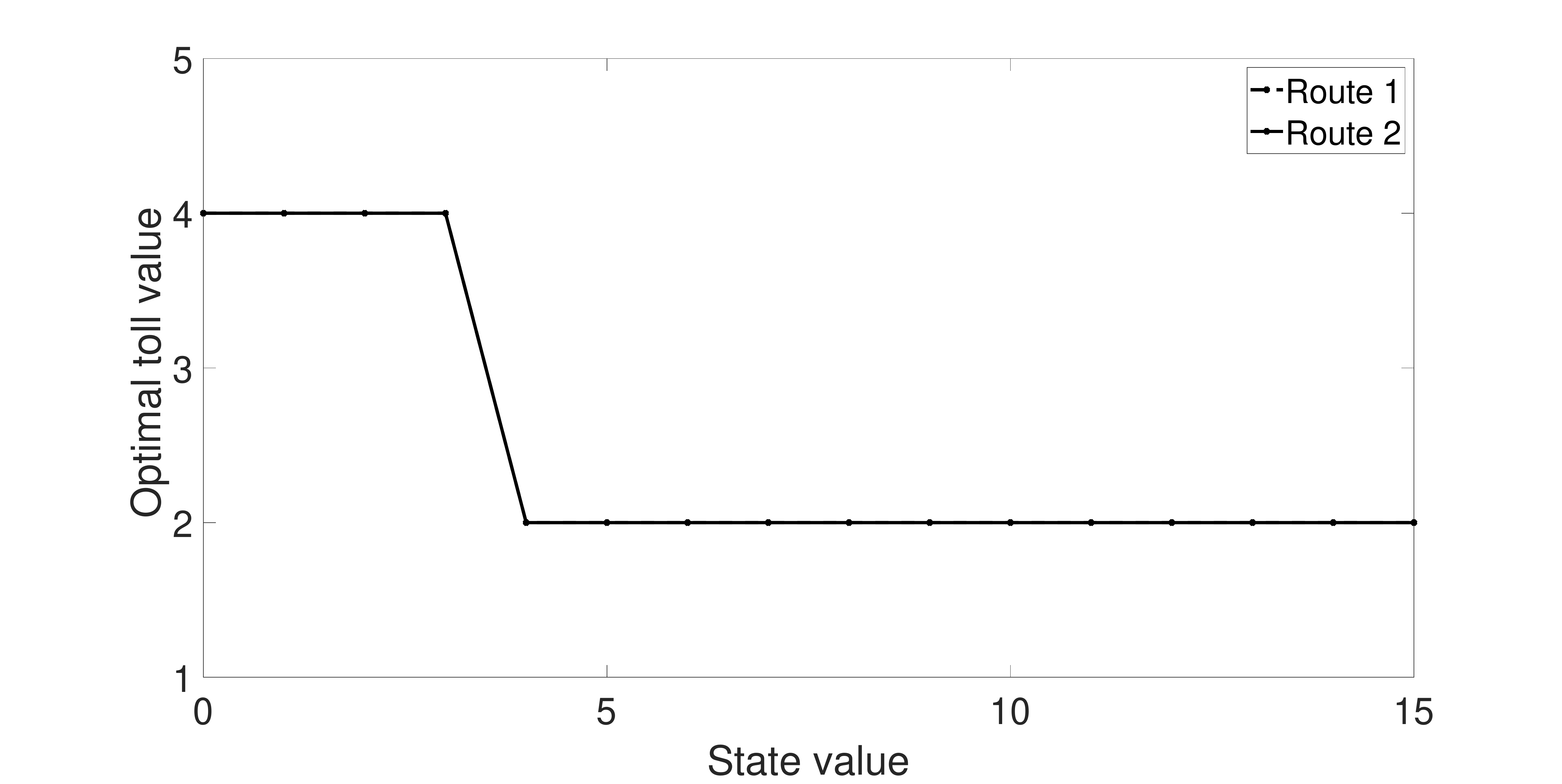}
	\caption{Plot of optimal toll values for different states by solving the \textit{original} problem.}
	\label{fig:soln}
\end{figure}
Since the expressions \eqref{eq:upperbound_g_optimality} and \eqref{eq:upperbound_g_optimality_2} only differ by a constant, the following discussion and conclusions are similar for both the expressions and therefore we only we focus on \eqref{eq:upperbound_g_optimality}. 
First, we consider the case when $\sqrt{\frac{k^{0,m}k^{1,l}\theta}{k^{1,m}}}-k^{0,l}x$ is positive. We plot the condition \eqref{eq:upperbound_g_optimality} in Figure \ref{fig:soln_costdiag_pos}. If the current demand is low (i.e., $k^{0,l}x<\sqrt{\frac{k^{0,m}k^{1,l}\theta}{k^{1,m}}}$) then the optimal toll values satisfy Equation \eqref{eq:upperbound_g_optimality} as the LHS of \eqref{eq:upperbound_g_optimality} is always positive. For a given set of the parameters, optimal toll value increases as the demand decreases by Equation \eqref{eq:upperbound_g_optimality}. In the extreme case, where the length of perpendicular from the origin to the line denoting Equation \eqref{eq:upperbound_g_optimality} becomes very large, i.e., $\frac{\sqrt{\frac{k^{0,m}k^{1,l}\theta}{k^{1,m}}}-k^{0,l}x}{\sqrt{(k^{1,l})^2+(k^{2,l})^2}}>\sqrt{(\tau^{max})^2+(\tau^{max})^2}$, then it is optimal to put the largest possible toll value $\tau^{max}$ so that the difference between the LHS and RHS of the Equation \eqref{eq:upperbound_g_optimality} is minimum. We now consider the case when $\sqrt{\frac{k^{0,m}k^{1,l}\theta}{k^{1,m}}}-k^{0,l}x$ is not positive as shown in Figure \ref{fig:soln_costdiag_neg}. Since the parameters and the possible toll values are positive, the LHS cannot be negative and thus it is optimal to put the least toll value $\tau^{min}$ so that the difference between the LHS and RHS is minimum. Thus, the optimal toll values are governed by the interplay between the demand $x$, $\theta$ and $TSTT(x,u)$ parameters. Note that $g\left(x,u\right)\ge  k^{0,l}\left(\frac{\theta}{k^{0,m}x+\sum_{r }k^{r,m} u^r}\right)+\sum_{r }k^{r,l} u^r$ (this can be obtained by performing an analysis that is similar to Lemma \ref{lem:expectedvalue_tstt}). Conducting a similar analysis on the partial derivative of the expression $k^{0,l}\left(\frac{\theta}{k^{0,m}x+\sum_{r }k^{r,m} u^r}\right)+\sum_{r }k^{r,l} u^r$ as before results into same conclusions regarding the optimal tolling solutions and therefore we omit the analysis for the sake of brevity. Finally, in the limiting case where the difference between the upper and lower bounds of $g(x,u)$ becomes small, it is possible to extend the same conclusions to $g(x,u)$.   
\begin{figure}[h!]
\centering
	\includegraphics[width= 0.6\textwidth]{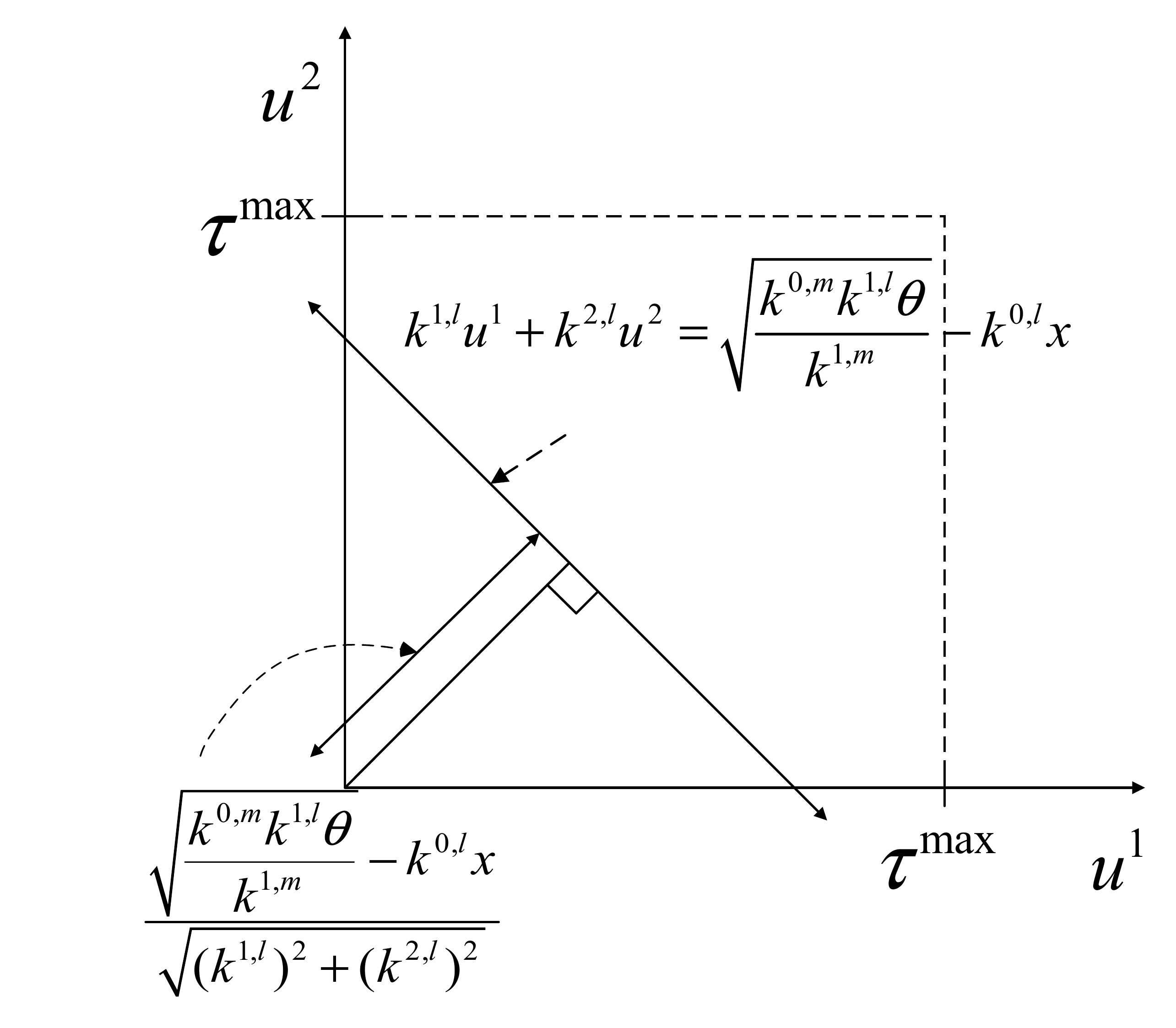}
	\caption{Plot of condition \eqref{eq:upperbound_g_optimality} when $\sqrt{\frac{k^{0,m}k^{1,l}\theta}{k^{1,m}}}-k^{0,l}x$ is positive. }
	\label{fig:soln_costdiag_pos}
\end{figure}
\begin{figure}[h!]
\centering
	\includegraphics[width= 0.6\textwidth]{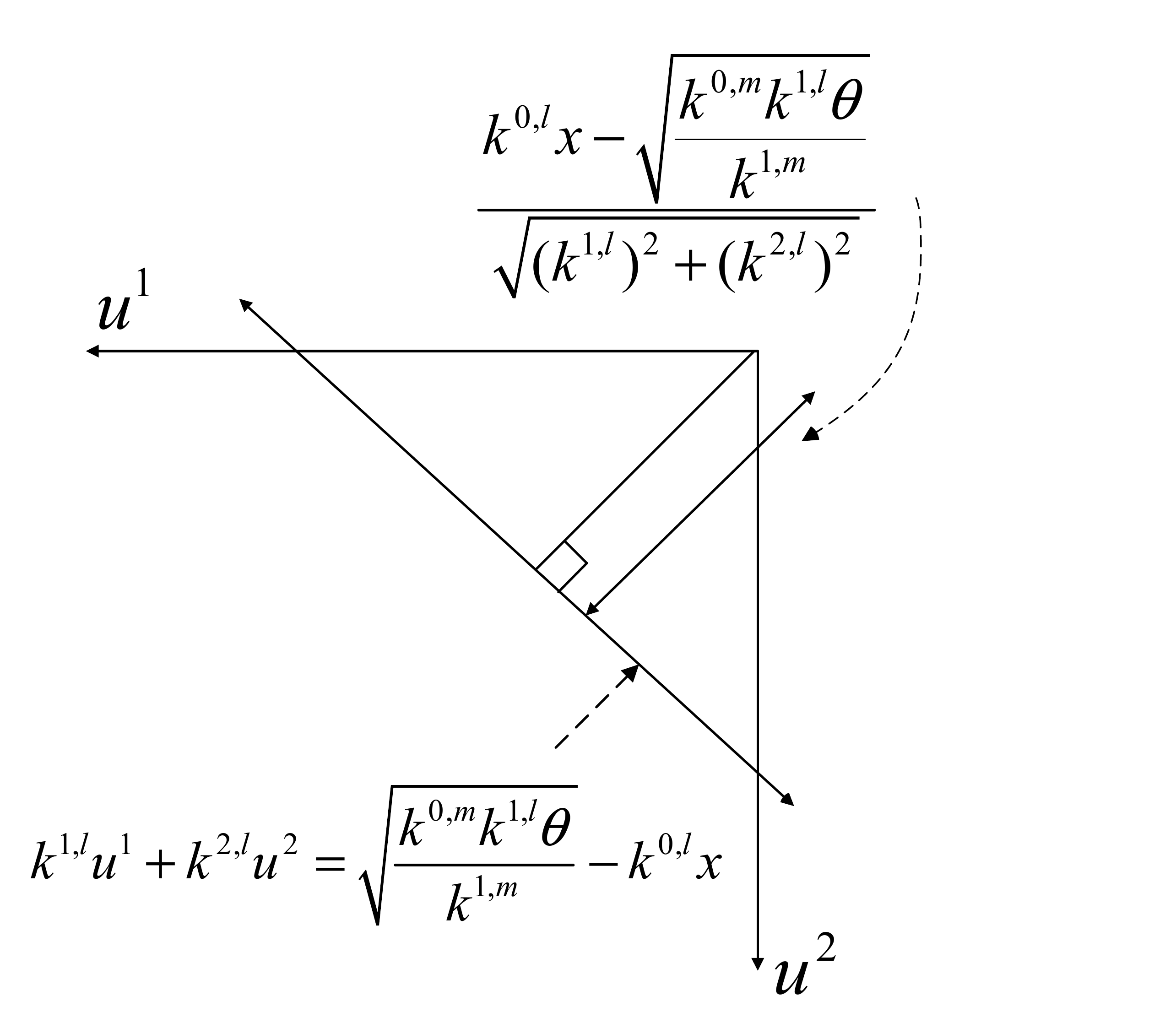}
	\caption{Plot of condition \eqref{eq:upperbound_g_optimality} when $\sqrt{\frac{k^{0,m}k^{1,l}\theta}{k^{1,m}}}-k^{0,l}x$ is not positive. }
	\label{fig:soln_costdiag_neg}
\end{figure}
\begin{figure}[h!]
\centering
	\includegraphics[width= 1.08\textwidth]{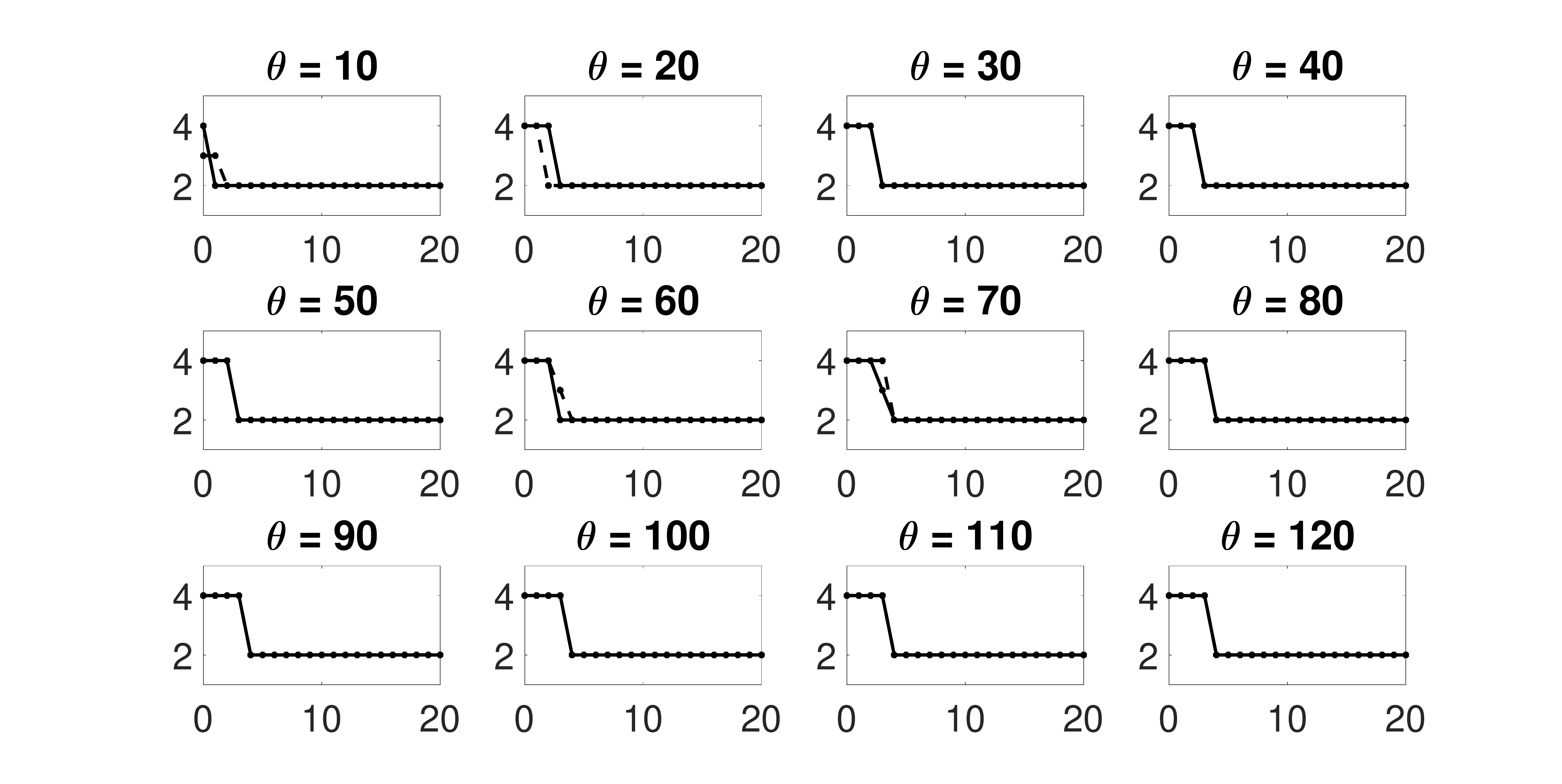}
	\caption{Plot of optimal toll values by solving the \textit{original} problem for different values of $\theta$. The x-axis represents state value (demand value) and y-axis represents the corresponding optimal toll value. The dashed and solid lines present optimal toll values for routes 1 and 2, respectively.}
	\label{fig:soln_thetavar}
\end{figure}

We now present some numerical results to illustrate the point in the last paragraph. Figure \ref{fig:soln_thetavar} presents the variation of optimal toll values with $\theta$, by keeping the travel time parameters to be fixed. It can be seen that as $\theta$ increases the optimal toll value for a given demand value also increases. This can also be followed from the analysis of bounds of the expected cost $g(x,u)$ in the last paragraph: as $\theta$ increases, the toll values $u^1$ and $u^2$ also increase by Equation \eqref{eq:upperbound_g_optimality}, for a given demand value and travel time parameters. This is also what one would anticipate: when $\theta$ is low, the toll should be low, since the demand in the next time step is going to be small with large probability by Equation \eqref{eq:poissondist}; when $\theta$ is large, the demand in the next time step is likely to be large, and therefore charging a large toll is appropriate. Thus, optimal toll calculation is anticipatory in nature as it takes into account the possible demand distribution of the next time step.     
\begin{figure}[h!]
\centering
	\includegraphics[width= 0.8\textwidth]{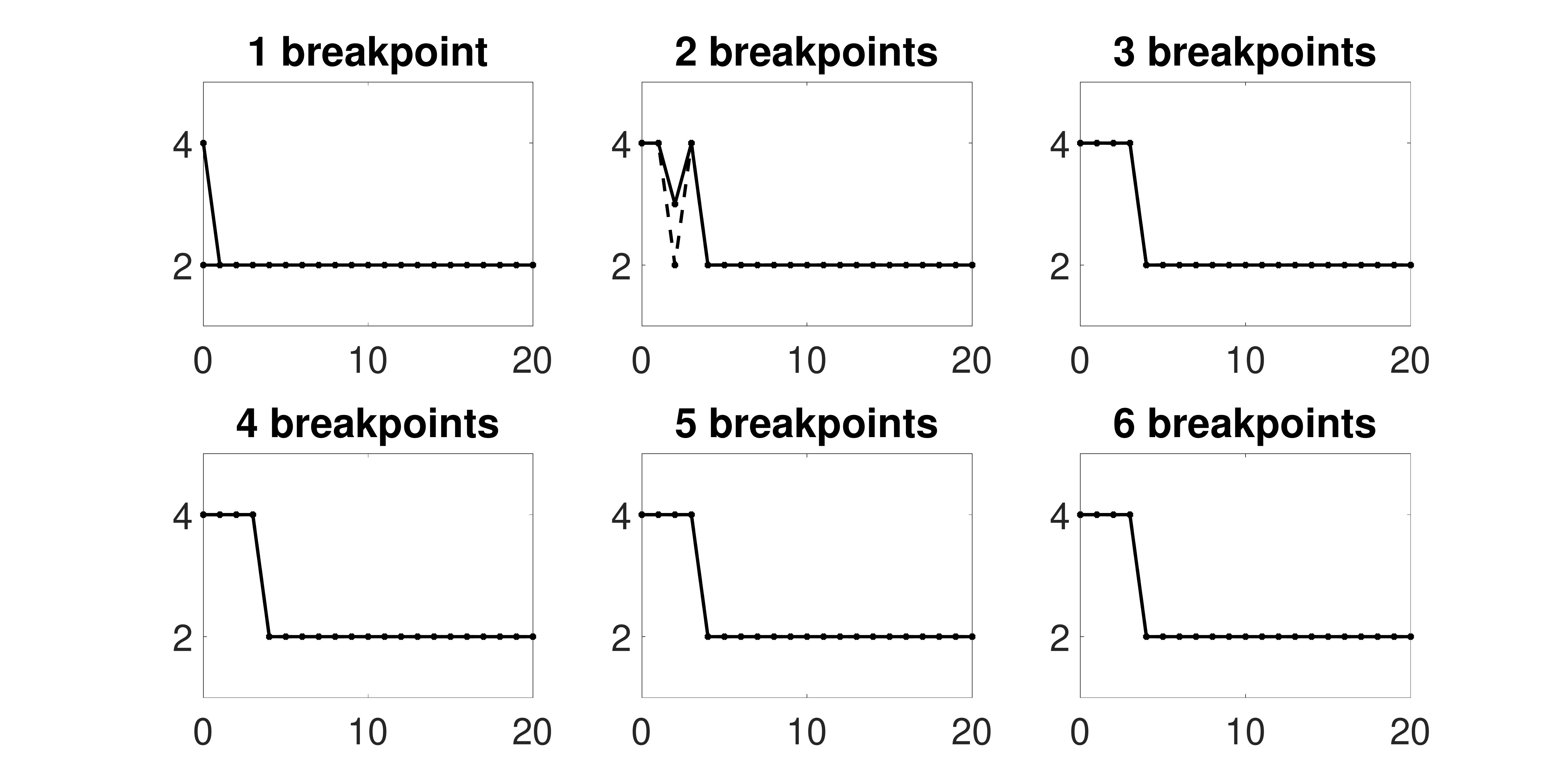}
	\caption{Plot of optimal toll values by solving the \textit{original} problem for different values of number of breakpoints (equivalently the number of piece-wise linear segments) of BPR functions. The x-axis represents state value (demand value) and y-axis represents the corresponding optimal toll value. The dashed and solid lines present optimal toll values for routes 1 and 2, respectively.}
	\label{fig:soln_numbofpiecesegments}
\end{figure}
\begin{figure}[h!]
\centering
	\includegraphics[width= 0.75\textwidth]{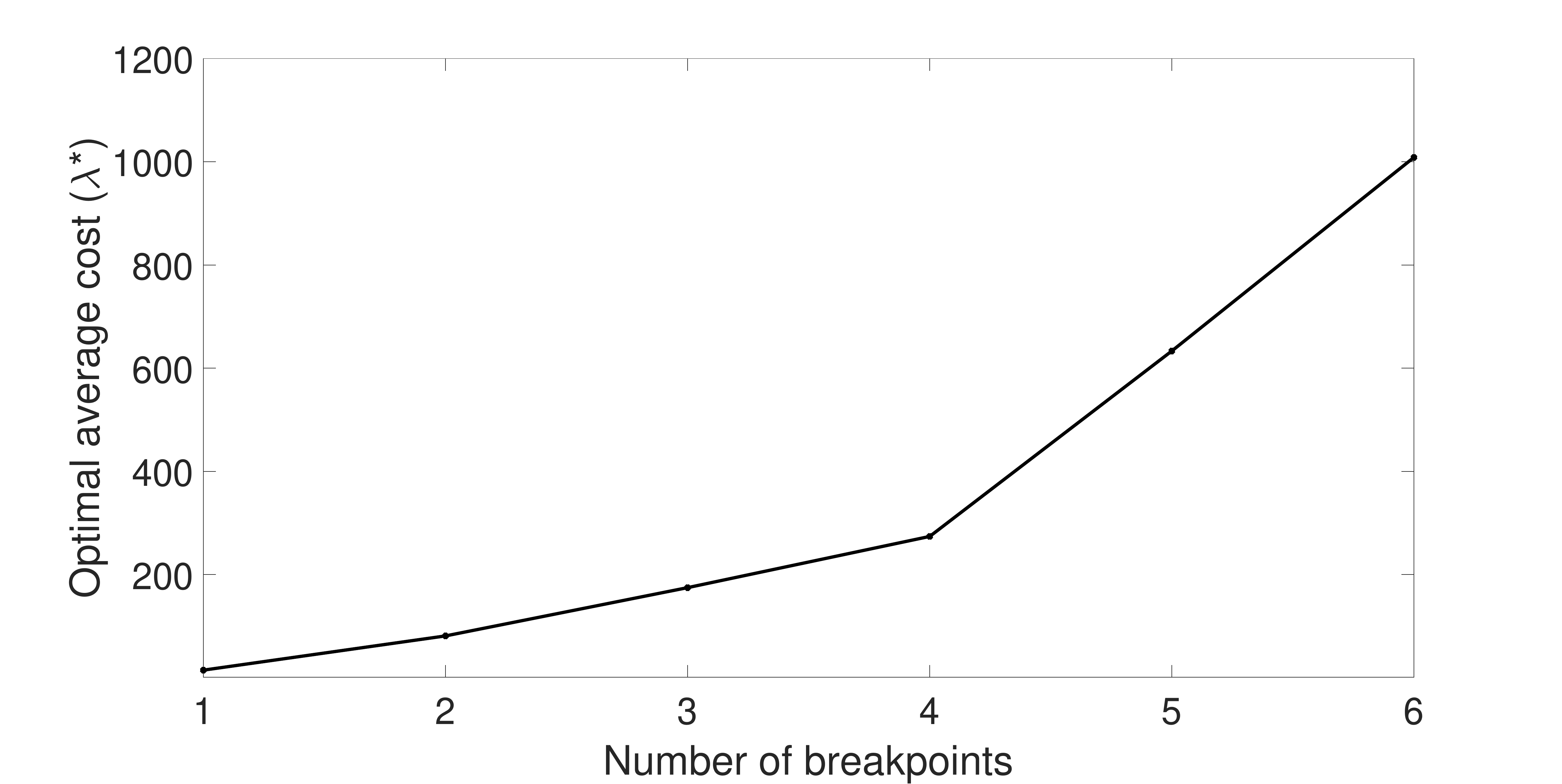}
	\caption{Variation of optimal average cost with the number of breakpoints. }
	\label{fig:cost_numbofpiecewisesegments}
\end{figure}
\begin{figure}[h!]
\centering
	\includegraphics[width= 0.75\textwidth]{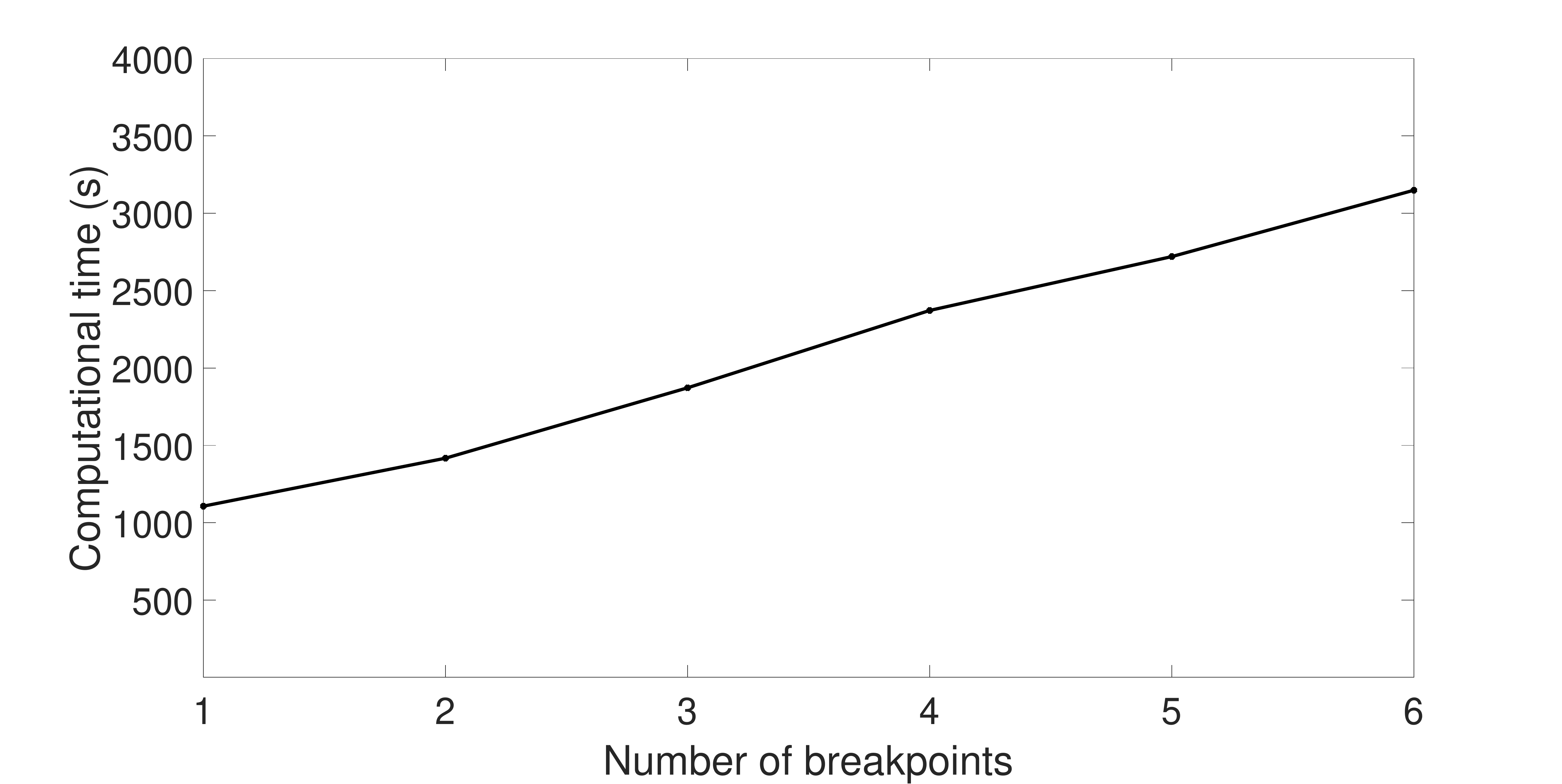}
	\caption{Variation of computational time (in seconds) with the number of breakpoints}
	\label{fig:comptime_numbofsegments}
\end{figure}
\subsection{Continuous piecewise linear approximation}
In this section, we discuss the effect of number of piece-wise segments (or equivalently the number of breakpoints) $\eta$ of BPR functions on the solution accuracy and computation. It is clear that as the number of piece-wise segments increase, the original BPR functions are more accurately approximated. However, the computational load involved in solving the problem also increases with $\eta$ and therefore there is a trade-off in choosing the value of $\eta$. Figure \ref{fig:soln_numbofpiecesegments} presents the variation of optimal toll values with $\eta$, by keeping the other parameters to be fixed. It can be seen that the optimal toll values do not vary when $\eta\ge 3$. Figure \ref{fig:cost_numbofpiecewisesegments} presents the variation of the optimal average cost $(\lambda^*)$ with $\eta$. It can be seen that the optimal average cost keeps increasing with $\eta$, which is also expected because as the number of segments increase the BPR travel time functions are more accurately approximated. Figure \ref{fig:comptime_numbofsegments} presents the variation of computational time with $\eta$. It can be seen that computational time steadily increases with $\eta$. Figures \ref{fig:soln_numbofpiecesegments},  \ref{fig:cost_numbofpiecewisesegments} and \ref{fig:comptime_numbofsegments} illustrate that though the accuracy of the optimal average cost increases with $\eta$, the increase in the computational time necessitates to choose an appropriate $\eta$. Since the optimal toll values do not change when $\eta$ becomes greater than three, we set $\eta=4$ in the \textit{original} problem.
\begin{figure}[h!]
\centering
	\includegraphics[width= 0.75\textwidth]{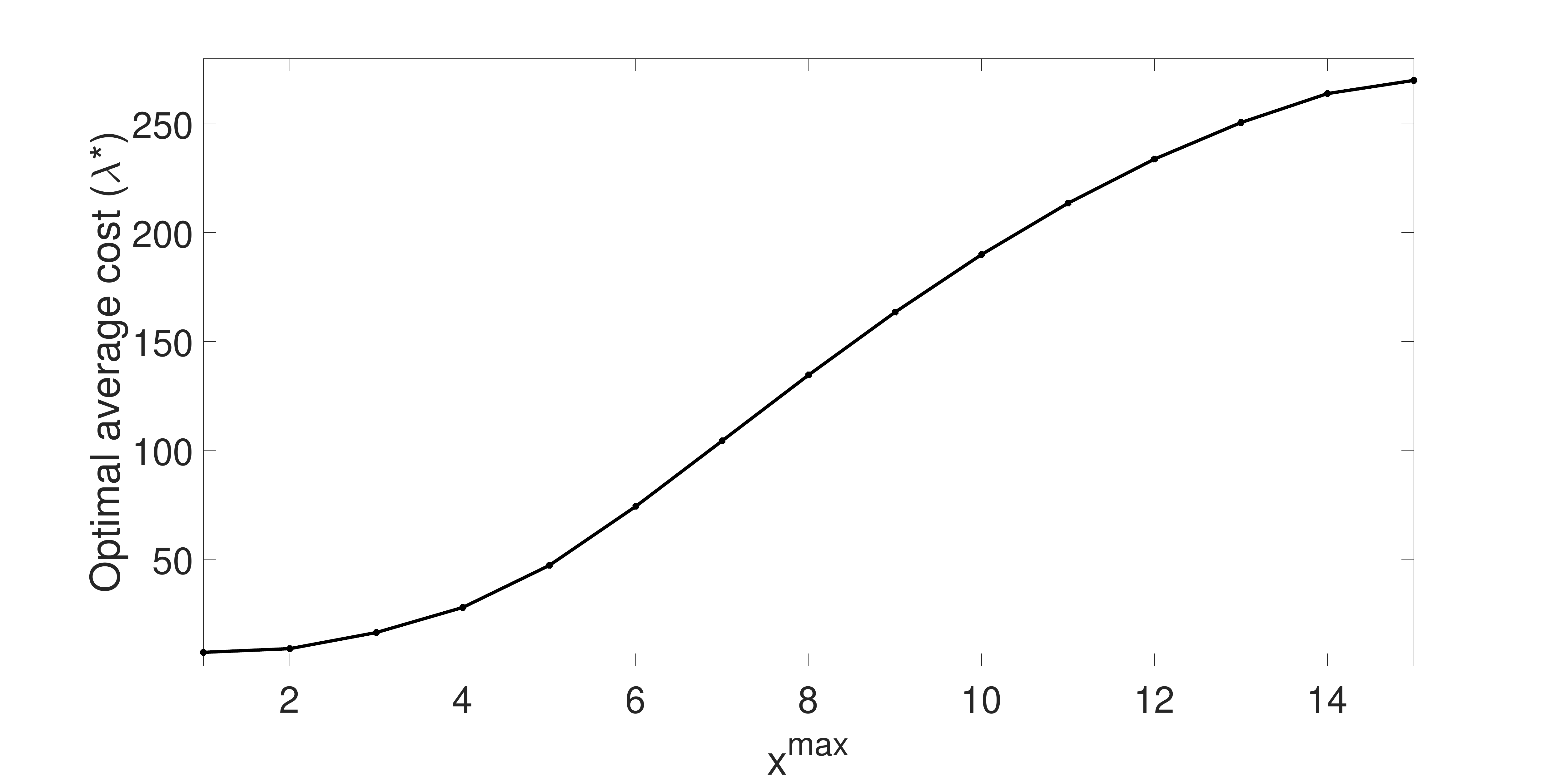}
	\caption{Variation of optimal average cost with the state truncation value ($x^{max}$). }
	\label{fig:maxvalue}
\end{figure}

\begin{figure}[h!]
\centering
	\includegraphics[width= 1.08\textwidth]{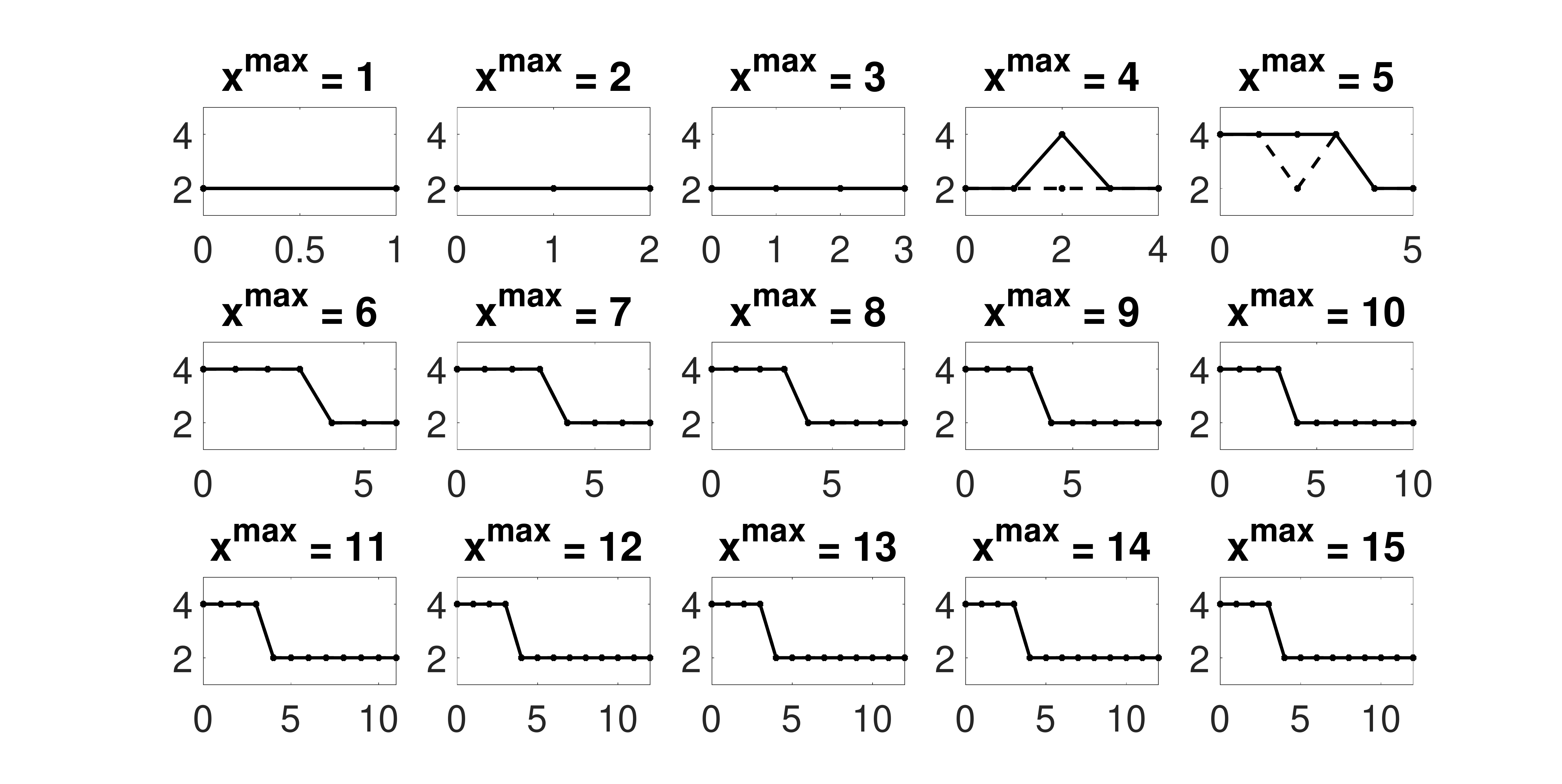}
	\caption{Plot of optimal toll values by solving the \textit{original} problem for different values of $x^{max}$. The x-axis represents state value (demand value) and y-axis represents the corresponding optimal toll value. The dashed and solid lines present optimal toll values for routes 1 and 2, respectively.}
	\label{fig:soln_xmaxvar}
\end{figure}
\subsection{Truncation}
In Section \ref{sec:implementation}, we introduced truncation of the states as a strategy to solve the infinite state problem. We said that as the truncation value ($x^{max}$) increases, the solution of the truncated problem can be more closely approximated to the solution of the untruncated problem. In Figure \ref{fig:maxvalue}, we plot the optimal average cost $\lambda^*$ obtained from solving the \textit{original} problem but for different values of $x^{max}$. It can be seen that $\lambda^*$ converges to a fixed value when $x^{max}$ approaches 15. Figure \ref{fig:soln_xmaxvar} presents the variation of optimal toll values with $x^{max}$. It can be seen that optimal toll values are the same for $x^{max}\ge 6$. Thus, choosing $x^{max}$ equal to 15 seems reasonable. 
Also, the choice of a larger truncation value $x^{max}$ comes at the expense of a significantly increased computation time (to be shown in the next section). Thus, there is a trade-off between solution quality versus computation time when we make the choice of truncation value.  

\subsection{Variation of computation time with the number of states}
In Section \ref{sec:aggregated}, we stated that as the number of states in the problem increase the computational time for solving the problem significantly increases. Figure \ref{fig:comptime_maxvalue} presents the computational time  when the \textit{original} problem is solved by varying the values of $x^{max}$. Note that the number of states in the problem are equal to $1+x^{max}$ because the state values start from zero. Thus, increasing the value of $x^{max}$ implies increasing the number of states in the problem. Since the computational time increases steadily with the number of states, truncated problem needs to be approximated so that the solutions of the problem can be efficiently solved. Therefore, from this point we present the results by solving an approximation of the \textit{original} problem.
\begin{figure}[h!]
\centering
	\includegraphics[width= 0.75\textwidth]{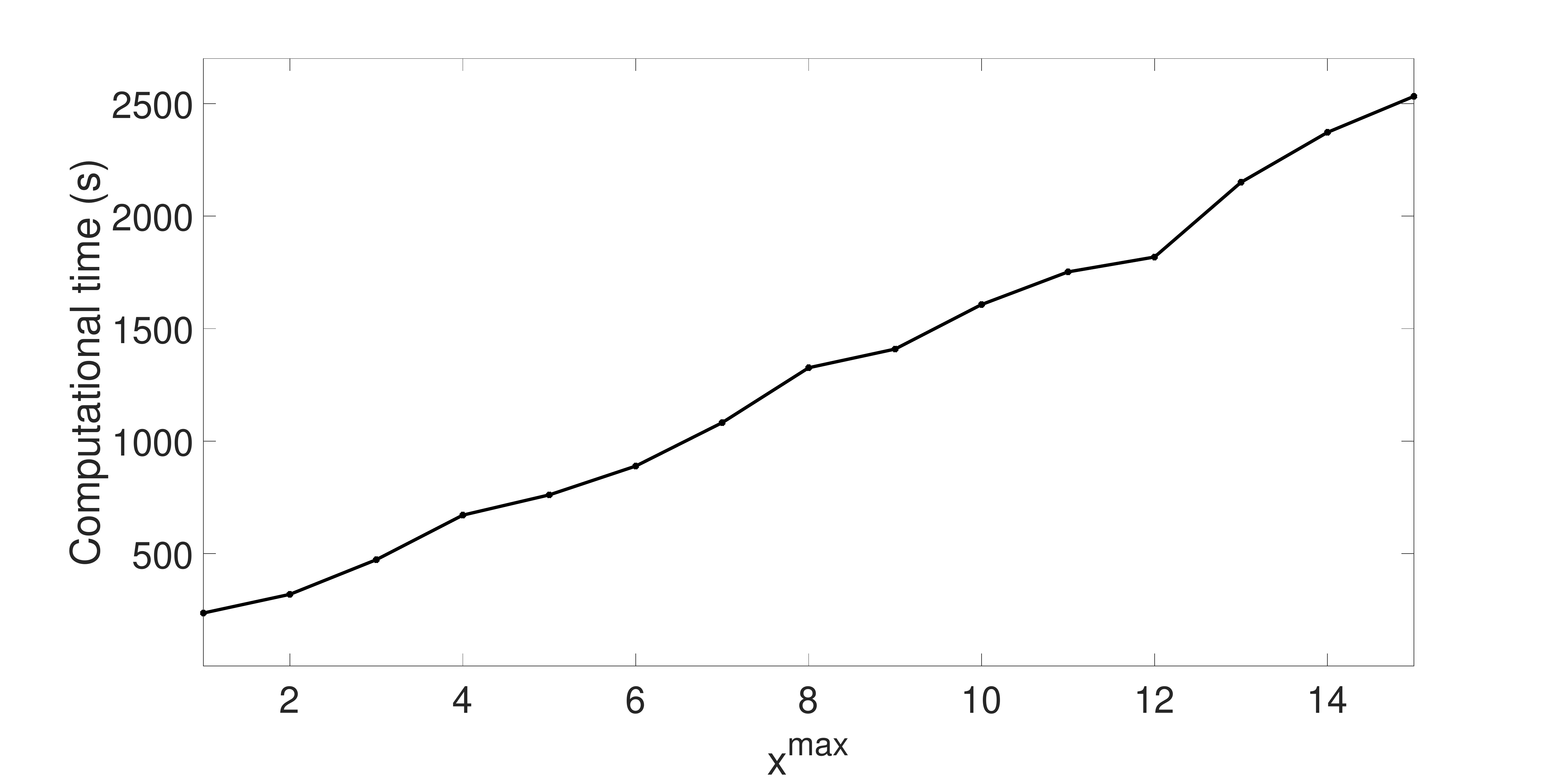}
	\caption{Variation of computational time (in seconds) with the state truncation value ($x^{max}$). }
	\label{fig:comptime_maxvalue}
\end{figure}
\begin{figure}[h!]
\centering
	\includegraphics[width= 0.75\textwidth]{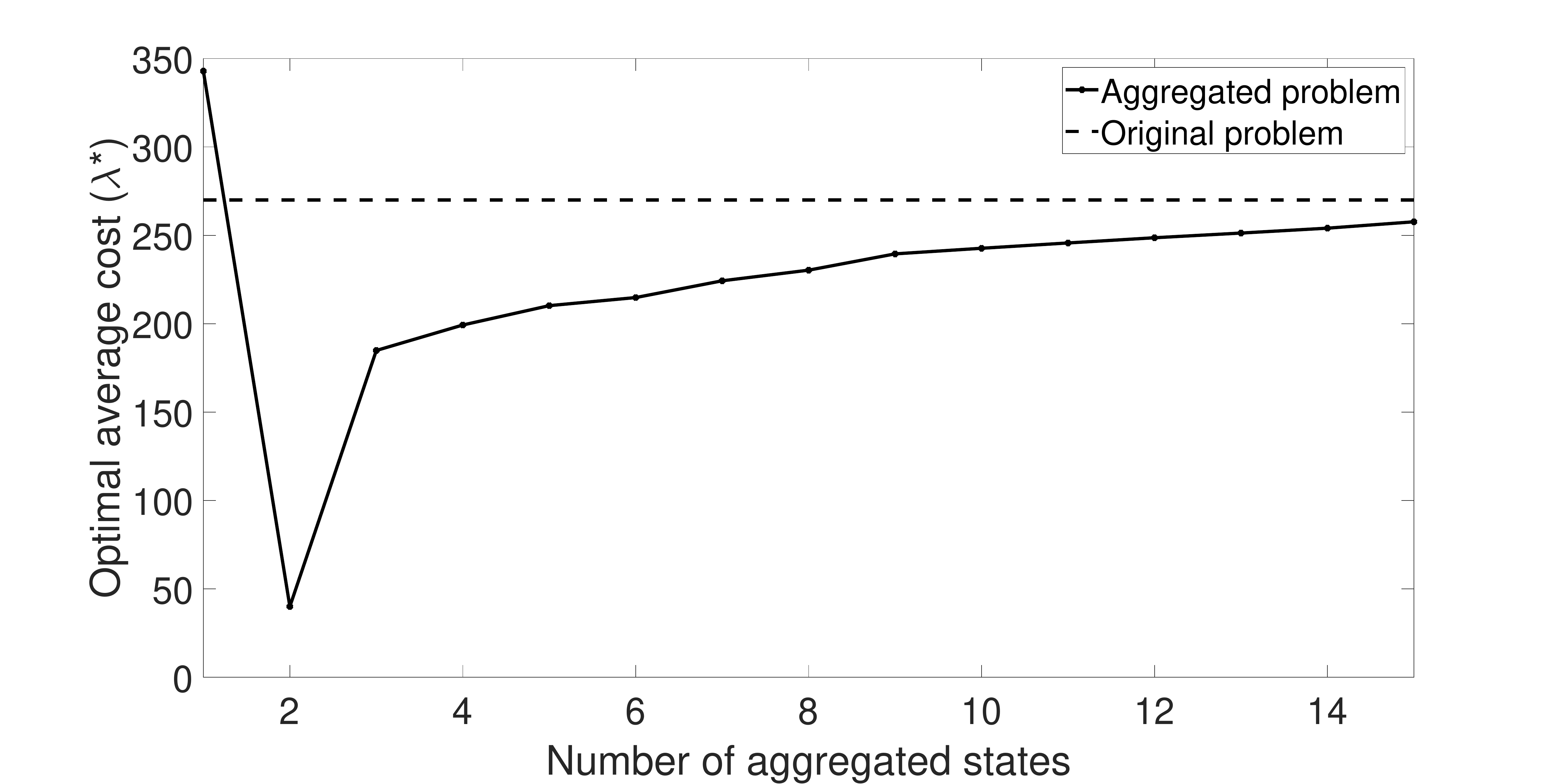}
	\caption{Variation of optimal average cost with the number of aggregated states.}
	\label{fig:cost_aggregation}
\end{figure}

\subsection{Variation of optimal average cost with the number of aggregated states}
Before we perform aggregation of the problem as suggested in the last section, we compare the solution quality of the aggregated problem with the \textit{original} problem. Figure \ref{fig:cost_aggregation} presents the comparison of the optimal average cost obtained from solving the \textit{original} and aggregated problems. The points on the solid line in the figure present the value of optimal average cost for various number of aggregated states in an aggregated problem. The dashed line represents the optimal average cost obtained from solving the \textit{original} problem. It can be seen that solutions of aggregated problems approach to the solution of the \textit{original} problem as the number of aggregated states increase. Based on these tests, we decided to solve the subsequent results with 8 aggregated states. Note that the computational time associated with solving the problem with 8 aggregated states is about 1204 seconds, which is about 0.47 times of the computational time involved in solving the \textit{original} problem.



\subsection{Variation of computational time with the number of routes}
In the \textit{original} problem, we considered two routes connecting the OD pair. We analyze the variation of computational time for solving the problem as the number of routes increase. Figure \ref{fig:numb_routes} presents this variation when the aggregated problem is solved. It can be seen that computational time significantly increases with the number of routes. This is in accordance with network modeling problems that face increasing computational times with network size \citep{gehlot2019user}. A possible remedy to tackle this issue in future works can be to aggregate routes similarly as we aggregate states in the MDP.   
%
\begin{figure}[h]
\centering
	\includegraphics[width= 0.65\textwidth]{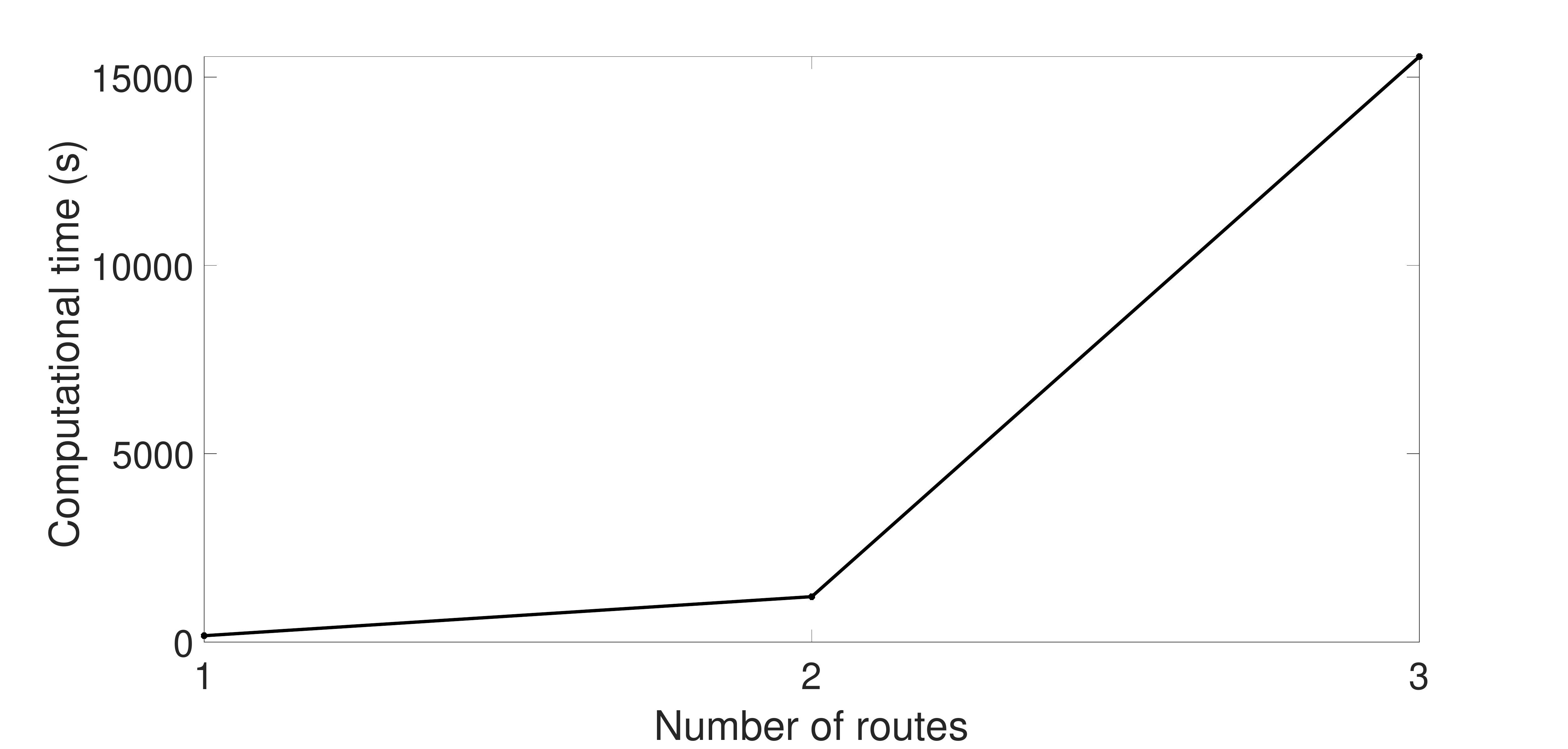}
	\caption{Variation of computational time (in seconds) with the number of routes}
	\label{fig:numb_routes}
\end{figure}

\section{Conclusions and future directions}
In this paper, we propose an optimal control modeling approach for day-to-day timescale congestion pricing that incorporates demand elasticity and stochasticity. We develop a countable-state infinite-horizon MDP where we do not constrain the travel time functions to be bounded, requiring a more elaborate solution analysis centered around weighted sup-norm contractions and recurrence properties of Markov chains. We prove that the optimal average cost of the MDP exists by verifying sufficient conditions needed to ensure Bellman's optimality. 
We also develop an approximate method that resolves implementation and computational issues associated with solving the exact problem. By the conducted numerical tests, we find that the optimal tolling calculation is anticipatory of the demand distribution in the next time step and sometimes it is optimal to levy large tolls whereas sometimes it is optimal to levy low tolls, depending on the problem parameters. We also found that the approximate method is both efficient and accurate through numerical results.

There can be several extensions to this study. We assumed a stationary (time-invariant) model by keeping the parameter $\theta$ to be the same across different days. However, if we assume time-varying parameters then it is notably harder to analyze the non-stationary model; see \cite{bowerman1974} for the finite state space setting. Computing the optimal policy is complicated as well; in practice one may have to use rolling/planning horizon methods \citep{alden1992rolling} or use classic linear programming methods \citep{ghate2013linear} when the state space is finite. Indeed, there is very little literature on how to deal with non-stationary MDP’s with unbounded costs and countable state spaces - two critical features of our model. This promises to be a fruitful avenue for further research. Finally, considering stochasticity in the supply side will also be an interesting future study.  

\section*{Acknowledgement}
The authors are grateful to National Science Foundation for the grant CMMI 1520338 to support the research presented in this paper.
\appendix
\section{$\bf{TSTT(x,u)}$ is a continuous piece-wise linear  function} \label{app:tstt_piecewiselinear}
We argue that $TSTT(x,u)$ is a continuous piece-wise linear  function.
\begin{lemma} \label{lem:tstt_piecewise}
$TSTT(x,u)$ is a continuous piece-wise linear  function. 
\end{lemma}
\begin{proof}
Note that for a given value of the arguments $(x,u)$ there is a unique solution to \eqref{eq:tstt1_defn}-\eqref{eq:tstt4} in terms of the variables $\{x^r\}$ and $\{z^{r,\beta}\}$ as continuous piece-wise linear approximations of BPR travel time functions are monotonically increasing \citep{sheffi1985urban}.
Thus, depending on the values of $x^1,\ldots,x^R$, \eqref{eq:tstt_compl1} can be solved as a system of linear equations for all the used routes to obtain
\begin{equation}
w=  (R-\iota)x\frac{\underset{\forall r \notin \nu}{\Pi}c^r}{\underset{\forall r 
		\notin \nu}{\sum}\underset{\forall j \neq r}{\Pi}c^j}+(R-\iota)\frac{\underset{\forall r \notin \nu}{\sum} \left(\left(\underset{\forall j \neq r}{\Pi}c^j\right)(u^r+\alpha^r)\right)}{\underset{\forall r 
		\notin \nu}{\sum}\underset{\forall j \neq r}{\Pi}c^j},\label{eq:TSTTexpression_2}  
\end{equation}
where the coefficients $\{c^r\}$ and $\{\alpha^r\}$ are the corresponding parameters of the piece-wise approximations to the solutions $x^1,\ldots,x^R$ that are obtained from \eqref{eq:tstt1_defn}-\eqref{eq:tstt4}. Since the variables $x^1,\ldots,x^R$ are unique for a given set of arguments, the travel time parameters $\{c^r\}$ and $\{\alpha^r\}$ in \eqref{eq:TSTTexpression_2} are also unique. If the change in the arguments $(x,u)$ is such that after the change, $TSTT(x,u)$ is given by the same travel time parameters $\{c^r\}$ and $\{\alpha^r\}$ (i.e., if the change does not crosses a break point of $TSTT(x,u)$), then $TSTT(x,u)$ is linear with the input arguments by \eqref{eq:TSTTexpression_2} and \eqref{eq:tstt1_defn}. Therefore, non-linearity only occurs at the breakpoints of $TSTT(x,u)$.

We now focus on continuity. Note that $TSTT(x,u)$ is the sum of continuous piecewise-linear route functions such that the input parameters satisfy \eqref{eq:tstt1_defn}-\eqref{eq:tstt4}. \cite{sheffi1985urban} showed that complementarity conditions \eqref{eq:tstt_compl1}, \eqref{eq:tstt_compl2} and \eqref{eq:tstt3_2} are equivalent to a convex optimization problem with strictly convex objective function and linear constraints because the approximated travel functions are monotonically increasing. Since the optimal value of such a problem is continuous in the parameters $(x,u)$ \citep{bonnans2013perturbation}, $TSTT(x,u)$ is a continuous function. 
\end{proof}
\begin{lemma} \label{lem:tstt_mono}
$TSTT(x,u)$ is a monotonically increasing function with the input parameters. 
\end{lemma}	
\begin{proof}
We first focus on the monotonicity of $TSTT(x,u)$ with state $x$ and toll values on the used routes. $TSTT(x,u)$ monotonically increases with these parameters because $TSTT(x,u)$ is a continuous piece-wise linear function by Lemma \ref{lem:tstt_piecewise} and the coefficients correponding to $x$ and the toll values on used routes are positive in-between the breakpoints by \eqref{eq:TSTTexpression_2}. Finally, the monotonicity of $TSTT(x,u)$ with toll values on the unused routes is implied by \eqref{eq:tstt1_defn}. 
\end{proof}

\section{Extension to networks with overlapping routes} \label{app:overlapping}
In Section \ref{subsec:tstt}, we derived the expression of total system travel time $(TSTT)$ for a single OD pair with non-overlapping routes (consisting of single links) in terms of various parameters. We now analyze the problem with overlapping routes. Notice that in order to conduct the analysis in Section \ref{sec:solnmethod}, we used the fact that total system travel time is a continuous piece-wise linear  function, i.e., $TSTT(x,u)=\sum_{\psi=1}^\xi \left(k^{0,\psi}x+\sum_{r }k^{r,\psi} u^r\right)\gamma^{\psi}$, where $x$ is the demand of the OD pair and $\{u^r\}$ is the set of toll values on various routes (equivalently links because each route had a single link then) in the network. We now analyze if the same form holds when there are routes with multiple links such that there can be overlaps between routes in a network. We present an algorithm that can be used to argue that the expression of $TSTT$ for single OD pair networks with overlapping routes is also a continuous piece-wise linear function of the demand and toll values on various links. Note that a network is a composition of two types of connections: \textit{series}  and \textit{parallel}. A set of elements (e.g. links) are connected in \textit{series} if the same amount of traffic flows through them. A set of elements in the network are said to be in \textit{parallel} if they are connected between the same pair of nodes. In Figure \ref{fig:overlapping_network}, consider a network with origin A, destination D and two overlapping routes connecting the OD pair: 1) link 1-link 2-link 4, and 2) link 3-link 4. In this network, the same traffic flows through links 1 and 2, so they are connected in series. Link 3 and the combined set of links 1 and 2 are in \textit{parallel} as they are connected between the same pair of nodes, A and C. Finally, link 4 and the combined set of links 1, 2 and 3 are connected in \textit{series}. 
\begin{figure}[h]
\centering
	\includegraphics[width= 0.3\textwidth]{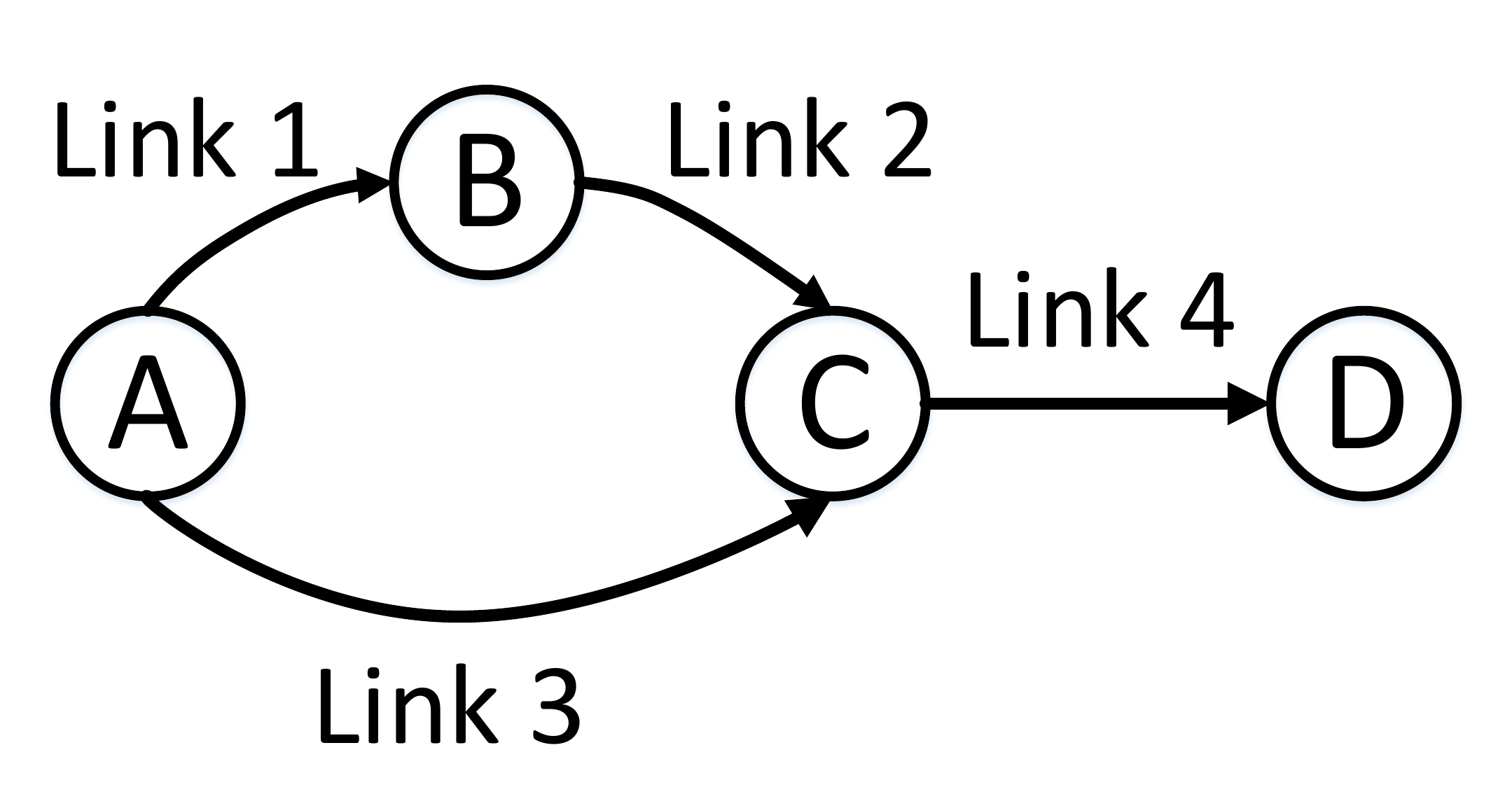}
	\caption{A network with overlapping routes}
	\label{fig:overlapping_network}
\end{figure}

We now present some results that we will use in computing the expression for $TSTT$.
\begin{lemma}\label{lem:jointlink_series}
Consider a pair of nodes that are connected by a route with $R$ consecutive links. The links have continuous piece-wise linear  travel time functions. Then, the total system travel time of this network is the same as that of a single link connecting the pair of nodes that has a continuous piece-wise linear travel time function. Thus, $R$ links of the route can be replaced by a single link that has a positive continuous piece-wise linear travel function. 
\end{lemma}
\begin{proof}
Denote $x$ as the travel demand flowing between the pair of nodes. The flow is positive and equal across all the links as there is a single route connecting the OD pair. Note that link travel time functions are continuous piece-wise linear in the flow and linear in the corresponding toll values. Therefore, the sum of travel times across the $R$ links is also continuous piece-wise linear in the flow and linear in the toll values. Thus, we can replace the $R$ links by a single link whose travel time function is continuous piece-wise linear in the arguments (note that a linear function is also a continuous piece-wise linear function) and is given by the sum of the travel times across the $R$ links.
\end{proof}

\begin{lemma} \label{lem:jointlink_parallel}
Consider a pair of nodes that have $R$ non-overlapping routes (such that each route has a single link) connecting the OD pair. The routes have continuous piece-wise linear  travel time functions. Then, the $R$ routes can be replaced by a single link that has a continuous piece-wise linear travel time function with positive coefficients and the travel time of this link is equal to the total system travel time of the network consisting of $R$ routes.
\end{lemma}
\begin{proof}
This proof of this result follows from the Lemmas \ref{lem:tstt_piecewise} and \ref{lem:tstt_mono}, which say that the total system travel time of links connected in parallel is a continuous piece-wise linear function. Thus, we can replace the $R$ routes with a single link that has travel time function equal to the total system travel time of the network consisting of $R$ routes.
\end{proof}
We now present the algorithm to reduce a network with overlapping routes to a single link network that has the same $TSTT$ as the original network:
\begin{enumerate}
\item Denote U and V as the upstream and downstream nodes in the current iteration, respectively.
\item Initialize the upstream node as the origin of the network.  
\item Pick an outgoing route from node U. Move along this route until a node is found that is an intersection of multiple routes. Denote this node as V. 
\item Find all the elements that are connected to the nodes U and V in parallel. Using Lemmas \ref{lem:jointlink_series} and \ref{lem:jointlink_parallel}, replace all the elements such that only single links are connected in parallel between the nodes U and V. 
\item Replace all the parallel links joining U and V with a single link by Lemma \ref{lem:jointlink_parallel}. 
\item Replace all the links (in series) connecting the origin and node V with a single link by Lemma \ref{lem:jointlink_series}.  
\item Set U equal to V. If U is the destination node then stop. Otherwise, go to step 3. 
\end{enumerate}
The correctness of the above algorithm can be proved by a loop invariant scheme, with the loop invariant being \textit{at a particular iteration of the algorithm, there is a single link with a linear piece-wise continuous travel time function with positive parameters connecting the origin and U.} 

\begin{figure}[h]
\centering
	\includegraphics[width= 0.6\textwidth]{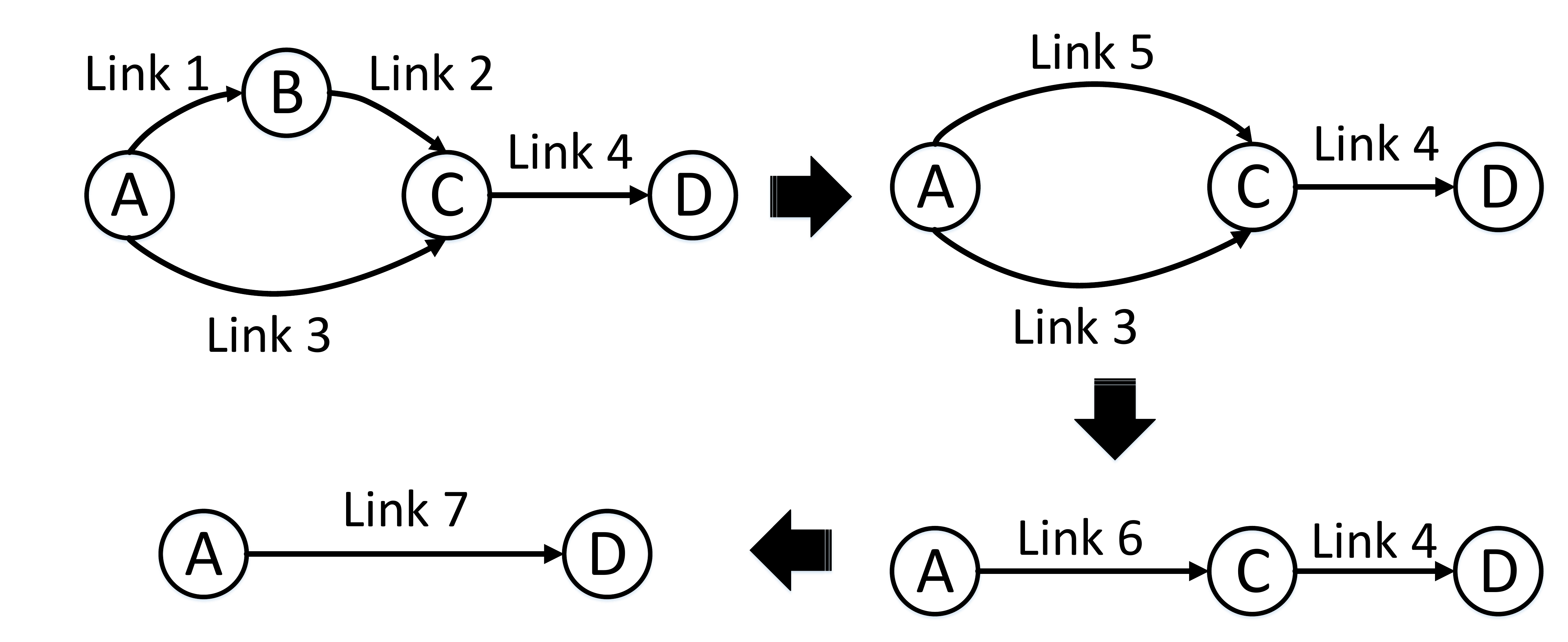}
	\caption{Reduction of an example network with overlapping routes to a single link network}
	\label{fig:overlapping_network_reduction}
\end{figure}

To demonstrate how this algorithm works, consider Figure \ref{fig:overlapping_network_reduction}: 1) First, initialize U as node A. Then, node C is node V. Then, we replace links 1 and 2 with a single link, say link 5 using Lemma \ref{lem:jointlink_series}. Now links 5 and 3 are connected in parallel. So, links 5 and 3 can be replaced by a single link, say link 6 by Lemma \ref{lem:jointlink_parallel}. In the next iteration, nodes C and D become the U and V nodes, respectively. Finally, links 6 and 4 can be replaced by a single link (say link 7) using Lemma \ref{lem:jointlink_series}. So, a single link with travel time of the form $\sum_{\psi=1}^\xi \left(k^{0,\psi}x+\sum_{r }k^{r,\psi} u^r\right)\gamma^{\psi}$,  connects the nodes A and D. Thus, we get the expression for total system travel time for this network. We can see that by using Lemmas \ref{lem:jointlink_series} and \ref{lem:jointlink_parallel}, the travel time parameters and the coefficients corresponding to toll values remain positive. Thus, the constants $k^{0,\psi}$ and $\{k^{r,\psi}\}$ in the expression $TSTT=\sum_{\psi=1}^\xi \left(k^{0,\psi}x+\sum_{r }k^{r,\psi} u^r\right)\gamma^{\psi}$ for networks with overlapping routes are positive. Thus, overlapping networks have $TSTT$ of the same form as non-overlapping networks and the analysis of sufficiency conditions that we earlier performed carries over. Finally, the action space in the problems with overlapping routes is the vector of link tolls and thus the proposed tolling scheme is not restricted to routes.

\section{Extension to multiple OD pair networks}\label{app:multipleOD}
We now demonstrate that the MDP model carries over to multiple OD-pair networks (that can possibly have overlapping routes). We consider a special case of multiple OD pairs, which we term as \textit{series activity trips} networks. Consider the network in Figure \ref{fig:multipleOD_tripactivity} with three OD pairs: $O$-$D_1$,$O$-$D_2$ and $O$-$D_3$. At a given time step, $x$ is the number of people traveling from node $O$ to different destinations: $\rho_1x$ people travel to node $D_1$; of the remaining $(1-\rho_1)x$ people, a fraction $\rho_2$ of them travel to node $D_2$ and the remaining $\left(1-\rho_1-\rho_2+\rho_1\rho_2\right)x$ people travel to node $D_3$. Such type of networks are relevant for modeling interrelated choices underlying trip chaining behavior. \cite{kitamura1984incorporating} expresses the destination choice decision that underlies trip chaining as a series of sequential choices. Suppose a traveler is about to make a visit to another location after completing a visit at the current location. It may happen that some travelers might not travel to the other location because of an unexpected incidence (e.g. coming across a friend) or the purpose of travel is satisfied at the current location. But for some travelers the purpose may not be satisfied at the current location (e.g. some stock is unavailable at the first location in a shopping trip, so some people may go to another shopping location), so they will travel to the next destination. Some recent studies also indicate that series trip chaining behavior happens in electric vehicles (EVs) due to different electric-charing opportunities along a route \citep{tamor2013statistical,xie2017path}. 

We characterize the state of the (controlled) system by the total demand across different OD pairs. The demand in each time step changes by the transition probability function of Equation \eqref{eq:poissondist}. 
\begin{figure}[h]
\centering
	\includegraphics[width= 0.8\textwidth]{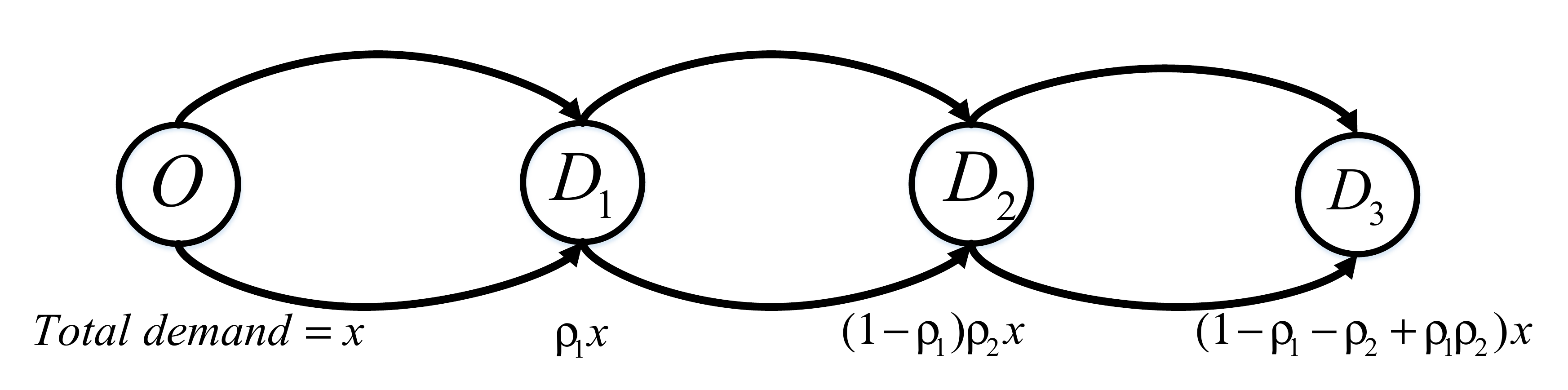}
	\caption{An example of a \textit{series activity trips} network}
	\label{fig:multipleOD_tripactivity}
\end{figure}
Consider a network with total demand at a given time step as $x$, then the demands across different OD pairs are some proportions of the total demand as shown in Figure \ref{fig:multipleOD_tripactivity}. Under this setting, we show that the total system travel time at the time step under consideration is given by the same form as before, i.e., $TSTT=\sum_{\psi=1}^\xi \left(k^{0,\psi}x+\sum_{r }k^{r,\psi} u^r\right)\gamma^{\psi}$, where $\{u^r\}$ is the vector of toll values on the links in the network. We again reduce the whole network to a single link that has travel time of the form $\sum_{\psi=1}^\xi \left(k^{0,\psi}x+\sum_{r }k^{r,\psi} u^r\right)\gamma^{\psi}$. This can be done through the following algorithm:
\begin{enumerate}
\item Consider an OD pair whose routes do not contain other origin or destination nodes in them. Initialize this OD pair as $OD_0$, where 0 denotes the iteration number. Then, all the routes connecting $OD_0$ can be replaced by a single link with continuous piece-wise linear  travel time function using the algorithm of Appendix \ref{app:overlapping}. This link has a travel time function that is continuous piece-wise linear  in terms of the flow traveling between $OD_0$ and the toll values on  links connecting the pair $OD_0$ (note that the flow traveling between $OD_0$ can also constitute demand from other OD pairs apart from the demand of $OD_0$). 

\item In iteration $n$ of the algorithm, we denote the OD pair under consideration as $OD_n$. Consider an OD pair, denoted as $OD_n$, that contains $OD_{n-1}$, i.e., if a user is traveling between $OD_n$ then he/she also travels between the pair $OD_{n-1}$. Then, $OD_{n-1}$ is connected by a single link that has travel time as a continuous piece-wise linear  function of the traffic that flows between the $OD_{n-1}$ (consisting of the demand of $OD_{n-1}$ and demands from other OD pairs). Consider the terms corresponding to the demands of $OD_{n-1},\ldots,OD_0$ as additional tolls in the travel time function of the link joining $OD_{n-1}$ as these do not constitute part of the demand between $OD_n$. Since Lemmas \ref{lem:jointlink_series} and \ref{lem:jointlink_parallel} ensure that the travel time functions of the new links are continuous piece-wise linear  functions of the flows as well as link toll values, the single link that replaces all the routes between $OD_{n}$ has travel time that is a continuous piece-wise linear  function of the demands corresponding to $OD_n,OD_{n-1},\ldots,OD_0$ as well as the toll values on the links connecting the pair $OD_{n}$.  
\item If the origin and destination nodes of the pair $OD_n$ do not have any incoming and outgoing links, respectively, then stop. Else, go to Step 2. 
\end{enumerate}
The proof of correctness of the above algorithm can be proceeded using the following loop invariant: \textit{after iteration $n$, all the routes joining $OD_n$ can be replaced by a single link that has travel time as a continuous piece-wise linear function of the demands of $OD_n,OD_{n-1},\ldots,OD_{0}$} as well as the toll values on the links connecting the OD pair $OD_{n}$. 

To demonstrate how the algorithm works, consider the network in Figure \ref{fig:overlapping_network} but with two OD pairs: $A$-$C$ and $A$-$D$ with demands at a given time step being $d^1=\rho x$ and $d^2=(1-\rho)x$, respectively, where $x$ is the total demand at the given time step. In the first iteration of the algorithm, $A$-$C$ is the OD pair under consideration. Then, links 1, 2 and 3 can be replaced by a single link 6 that has travel time function that varies piece-wise linearly and continuously with the total demand (i.e., the sum of the demands of both the OD pairs) and the toll values on links 1, 2 and 3 by the analysis in Appendix \ref{app:overlapping}. Denote the travel time of link 6 joining nodes $A$ and $C$ at the given time step by $\sum_{\psi=1}^\xi \left(k^{0,\psi}x+k^{1,\psi} u^1+k^{2,\psi} u^2+k^{3,\psi} u^3\right)\gamma^{\psi}$, where $\{k^{0,\psi}\},\{k^{1,\psi}\},\{k^{2,\psi}\},\{k^{3,\psi}\}$ are positive constants, $\{\gamma^{\psi}\}$ are binary variables and $u^1, u^2,u^3$ are tolls applied on links $1, 2$ and $3$ at the given time step, respectively. Alternatively, it can also be considered that link 6 has flow equal to $d^2$ flowing through it at the given time step with travel time parameter $k^{0,\psi}$ and toll equal to $k^{0,\psi}d^1+k^{1,\psi}u^1+k^{2,\psi}u^2+k^{3,\psi}u^3$ levied on it, given that $\psi$th segment of the continuous piece-wise function determines the value of $TSTT(x,u)$. Denote the travel time of link 4 at the given time step as $\sum_{\bar{\psi}=1}^{\bar{\xi}} \left({\bar{k}}^{0,\bar{\psi}}d^2+{\bar{k}}^{1,\bar{\psi}} u^4\right){\bar{\gamma}}^{\bar{\psi}}$, where $\{{\bar{k}}^{0,\bar{\psi}}\}$ and $\{{\bar{k}}^{1,\bar{\psi}}\}$ are positive constants, $\{{\bar{\gamma}}^{\bar{\psi}}\}$ are binary variables and $u^4$ is the toll value on link 4. Then, links 6 and 4 can be replaced by a single link by Lemma \ref{lem:jointlink_series} that is piece-wise linear and continuous in terms of the demand values of both the OD pairs and the toll values on different links. Since the demands across different OD pairs are fixed proportions of the total demand, the single link that connects $A$ and $D$ has travel time that is piece-wise linear and continuous in terms of the total demand across the network and the link toll values. So, for such networks $TSTT$ is of the same form as before and the solution analysis that we conducted for single OD pair networks carries over.

\bibliographystyle{apa}
\bibliography{mybibfile}

%
%

%

\end{document}